\documentclass[11pt,letterpaper]{article}
\usepackage{typearea}
\typearea{15}

\newif\iffullver
\fullverfalse
\fullvertrue %last one wins

\usepackage{theorem,latexsym,graphicx}
\usepackage{amsmath,amssymb,enumerate}
\usepackage{xspace}
\usepackage{bm,fullpage}
\usepackage{ifpdf}
\usepackage{shadow,shadethm,color}
\usepackage[compact]{titlesec}
\usepackage{algorithm}
\usepackage[noend]{algorithmic}
\usepackage{subfigure}
\usepackage{verbatim}
\usepackage{times}
\usepackage{paralist}

\definecolor{Darkblue}{rgb}{0,0,0.4}
\definecolor{Brown}{cmyk}{0,0.81,1.,0.60}
\definecolor{Purple}{cmyk}{0.45,0.86,0,0}
\newcommand{\mydriver}{hypertex}
\ifpdf
 \renewcommand{\mydriver}{pdftex}
\fi
\usepackage[breaklinks,\mydriver]{hyperref}
\hypersetup{colorlinks=true,%pdfborder={1 1 1 [3]},%
            citebordercolor={.6 .6 .6},linkbordercolor={.6 .6 .6},%
citecolor=Darkblue,urlcolor=black,linkcolor=Darkblue,pagecolor=black}
\newcommand{\lref}[2][]{\hyperref[#2]{#1~\ref*{#2}}}


\newtheorem{theorem}{Theorem}[section]
\newtheorem{definition}[theorem]{Definition}

\newtheorem{lemma}[theorem]{Lemma}
\newshadetheorem{lemmashaded}[theorem]{Lemma}

\newtheorem{claim}[theorem]{Claim}

\newtheorem{corollary}[theorem]{Corollary}

\numberwithin{algorithm}{section}

\newenvironment{proof}{

\noindent{\bf Proof:}}
{\hfill$\blacksquare$

}

\newcommand{\junk}[1]{}
\newcommand{\ignore}[1]{}

\newcommand{\Z}[0]{{\ensuremath{\mathbb{Z}}}}

\newcommand{\A}[0]{{\ensuremath{\mathcal{A}}}\xspace}
\newcommand{\lm}[0]{{\ensuremath{\mathcal{L}}}\xspace}
\newcommand{\N}[0]{{\ensuremath{\mathcal{N}}}\xspace}

\newcommand{\sse}{\subseteq}
\newcommand{\I}{{\mathcal{I}}}

\newcommand{\T}{\ensuremath{\mathcal{T}}\xspace}

\newcommand{\Opt}{\ensuremath{\mathsf{Opt}}\xspace}

\newcommand{\so}{\ensuremath{\mathsf{StocOrient}}\xspace}

\newcommand{\co}{\ensuremath{\mathsf{CorrOrient}}\xspace}
\newcommand{\kdo}{\ensuremath{\mathsf{KDO}}\xspace}

\newcommand{\lp}{\ensuremath{\mathsf{LP}}\xspace}
\newcommand{\dlp}{\ensuremath{\mathsf{DLP}}\xspace}

\newcommand{\size}{{S}}

\newcommand{\E}{\mathbb{E}}

\newcommand{\q}[1]{e\left(#1\right)}

\newcommand{\initOneLiners}{%
    \setlength{\itemsep}{0pt}
    \setlength{\parsep }{0pt}
    \setlength{\topsep }{0pt}
}

\begin{document}

\title{On the Adaptivity Gap of Stochastic Orienteering}
\author{
Nikhil Bansal\thanks{Eindhoven University of Technology. Supported by the Dutch NWO Grant 639.022.211.}
\and
Viswanath Nagarajan\thanks{IBM T.J. Watson Research Center}
}
\date{}

\maketitle

\begin{abstract}
The input to the {\em stochastic orienteering} problem~\cite{GKNR12} consists of a budget $B$ and metric $(V,d)$ where each vertex $v\in V$ has a job with a deterministic reward and a \emph{random} processing time (drawn from a known distribution). The processing times are independent across vertices. The goal is to obtain a non-anticipatory policy  (originating from a given root vertex)
to run jobs at different vertices, that maximizes expected reward, subject to the total distance traveled plus processing times being at most $B$. An {\em adaptive} policy is one that can choose the next vertex to visit based on observed random instantiations. Whereas, a {\em non-adaptive} policy is just given by a fixed ordering of vertices.
%: it visits and runs jobs in this order until the budget $B$ is exhausted.
The {\em adaptivity gap} is the worst-case ratio of the expected rewards of the optimal adaptive and non-adaptive policies.

\smallskip

We prove an $\Omega\left((\log\log B)^{1/2}\right)$ lower bound on the adaptivity gap of stochastic orienteering. This provides a negative answer to the $O(1)$-adaptivity gap conjectured in~\cite{GKNR12}, and comes close to the $O(\log\log B)$ upper bound proved there. This result holds even on a line metric.

\smallskip

We also show an $O(\log\log B)$ upper bound on the adaptivity gap for the {\em correlated} stochastic orienteering problem, where the reward of each job is random and possibly correlated to its processing time. Using this, we obtain an improved quasi-polynomial time $ \min\{\log n,\log B\}\cdot \tilde{O}(\log^2\log B)$-approximation algorithm for correlated stochastic orienteering.

\end{abstract}

\section{Introduction}

In the {\em orienteering} problem~\cite{GLV87}, we are given a metric $(V,d)$ with a starting vertex $\rho\in V$ and a budget $B$ on length. The objective is to compute a path originating from $\rho$ having length at most $B$, that maximizes the number of vertices visited. This is a basic vehicle routing problem (VRP) that arises as a subroutine in algorithms for a number of more complex variants, such as VRP with
time-windows, discounted reward TSP and distance constrained VRP.

The stochastic variants of orienteering and related problems such as traveling salesperson and vehicle routing have also been extensively studied. In particular, several dozen variants have been considered depending on which parameters are stochastic, the choice of the objective function, the probability distributions, and optimization models such as {\em a priori} optimization, stochastic optimization with recourse, probabilistic settings and so on.
For more details we refer to a recent survey \cite{Weyland} and references therein.

Here, we consider the following stochastic version of the orienteering problem defined by~\cite{GKNR12}. Each vertex contains a job with a deterministic reward and random processing time (also referred to as size); these processing times are independent across vertices. The processing times  model the random delays encountered at the node, say due to long queues or activities such as filling out a form, before the reward can be collected. The distances in the metric correspond to travel times between vertices, which are deterministic. The goal is to compute a {\em policy}, which describes a path originating from the root $\rho$ that visits vertices and runs the respective jobs, so as to maximize the total expected reward subject to the total time (for travel plus processing) being at most $B$.
Stochastic orienteering also generalizes the well-studied stochastic knapsack problem~\cite{DeanGV08,BGK11,Bhalgat11} (when all distances are zero).
We also consider a further generalization, where the reward at each vertex is also random and possibly {\em correlated} to its processing time. %(The instantiations at different vertices are still independent.)

A feasible solution (policy) for the stochastic orienteering problem is represented by a decision tree, where nodes encode the ``state'' of the solution (previously visited vertices and the residual budget), and branches denote random instantiations. Such solutions are called {\em adaptive} policies, to emphasize the fact that their actions may depend on previously observed random outcomes.
Often, adaptive policies can be very complex and hard to reason about.
For example, even for the stochastic knapsack problem an optimal adaptive strategy may have exponential size (and
several related problems are PSPACE-hard) \cite{DeanGV08}.

Thus a natural approach for designing algorithms in the stochastic setting is to:
(i) restrict the solution space to the simpler class of {\em non adaptive} policies (eg.~in our stochastic orienteering setting, such a policy is described by a fixed permutation to visit vertices in, until the budget $B$ is exhausted), and (ii) design an efficient algorithm to find a
(close to) optimum non-adaptive policy.

While non-adaptive policies are often easier to optimize over, the drawback is that they could be much worse than the optimum adaptive policy. Thus, a key issue is to bound the {\em adaptivity gap}, introduced by \cite{DeanGV08} in their seminal paper,  which is the worst-case ratio (over all problem instances) of the optimal adaptive value to the optimal non-adaptive value. %\cite{DeanGV08} showed that the adaptivity gap for the stochastic knapsack problem is $O(1)$.

In recent years, increasingly sophisticated techniques have been developed for designing good non-adaptive policies and for proving small adaptivity gaps~\cite{DeanGV08,GuhaM09,chenetal,BGLMNR10,GKMR11,GKNR12}.
% (see section \ref{s:prev-work} for previous work).
For stochastic orienteering, \cite{GKNR12} gave an $O(\log\log B)$ bound on the adaptivity gap, using an elegant probabilistic argument (previous approaches only gave a $\Theta(\log B)$ bound).
More precisely, they considered certain $O(\log B)$ correlated probabilistic events and used martingale tails bounds on suitably defined stopping times to bound the probability that none of these events happen.
In fact, \cite{GKNR12} conjectured that the adaptivity gap for stochastic orienteering was $O(1)$, suggesting that the $O(\log \log B)$ factor was an artifact of their analysis.
 %That paper left open the possibility of a constant-factor adaptivity gap for stochastic orienteering.

\subsection{Our Results and Techniques}

\noindent {\bf Adaptivity gap for stochastic orienteering:}
Our main result is the following lower bound.
\begin{theorem}\label{thm:ad-gap}
The adaptivity gap of stochastic orienteering is $\Omega\left((\log\log B)^{1/2}\right)$, even on a line metric.
\end{theorem}
This answers negatively the $O(1)$-adaptivity gap conjectured in~\cite{GKNR12}, and comes close to the $O(\log\log B)$ upper bound proved there. To the best of our knowledge, this gives the first non-trivial $\omega(1)$ adaptivity gap for a natural problem.

The lower bound proceeds in three steps and is based on a somewhat intricate construction.
 We begin with a basic instance described by a directed binary tree of height $\log\log B$ that essentially represents the optimal adaptive policy. Each processing time is a Bernoulli random variable: it is either zero, in which case the optimal policy goes to its left child, or a carefully set positive value, in which case the optimal policy goes to its right child.
The edge distances and processing times are chosen so that when a non-zero size instantiates, it is always possible to take a right edge, while the left edges can only be taken a few times. On the other hand, if the non-adaptive policy chooses a path with mostly right edges, then it cannot collect too much reward.

In the first step of the proof, we show that this directed tree instance has an $\Omega((\log\log B)^{1/2})$ adaptivity gap.
The main technical difficulty here is to show that every fixed path (which may possibly skip vertices, and gain advantage over the adaptive policy) either runs out of budget $B$ or collects low expected reward.
In the second step, we drop the directions on the edges and show that the adaptivity gap continues to hold (up to constant factors). The optimum adaptive policy that we compare against remains the same as in the directed case, and the key issue here is to show that
the non-adaptive policy cannot gain too much by backtracking along the edges.
To this end, we use some properties of the distances on edges in our instance.
In the final step, we embed the undirected tree onto a line at the expense of losing another $O(1)$ factor in  the adaptivity gap.
The problem here is that pairs of nodes that are far apart on the tree may be very close on the line.
To get around this, we exploit the asymmetry of the tree distances and some other structural properties to show that this has limited effect.

\medskip

\noindent  {\bf Correlated Stochastic Orienteering:}
Next, we consider the correlated stochastic orienteering problem, where the reward at each vertex is also random and possibly correlated with its processing time (the distributions are still independent across vertices). In this setting, we prove the following.% upper bound on the adaptivity gap.
\begin{theorem}\label{thm:corr-loglog-UB}
The adaptivity gap of correlated stochastic orienteering is $O(\log\log B)$.
\end{theorem}
 This improves upon the $O(\log B)$-factor adaptivity gap that is implicit in~\cite{GKNR12}, and matches the adaptivity gap upper bound known for uncorrelated stochastic orienteering. The  proof makes use of a martingale concentration inequality~\cite{Zhang05} (as~\cite{GKNR12} did for the uncorrelated problem), but dealing with the reward-size correlations requires a different definition of the stopping time. For the uncorrelated case, the stopping time~\cite{GKNR12} used a single ``truncation threshold'' (equal to $B$ minus the travel time)  to compare the instantiated sizes and their expectation.
In the correlated setting, we use $\log B$ different
truncation thresholds (all powers of $2$), irrespective of the travel time, to determine the stopping criteria.

\medskip

\noindent{\bf Algorithm for Correlated Stochastic Orienteering:}
Using some structural properties in the proof of the adaptivity gap upper bound above, we obtain an improved {\em quasi-polynomial}\footnote{A quasi-polynomial time algorithm runs in $2^{\log^c N}$ time on inputs of size $N$, where $c$ is some constant.} time algorithm for correlated stochastic orienteering.
\begin{theorem}\label{thm:corr-NA}
There is an $O\left( \alpha\cdot \log^2\log B/\log\log\log B\right)$-approximation algorithm for correlated stochastic orienteering, running in time $(n+\log B)^{O(\log B)}$. Here $\alpha\le \min\{O(\log n),\, O(\log B)\}$ denotes the best approximation ratio for the orienteering with deadlines problem.
\end{theorem}
The {\em orienteering with deadlines} problem is defined formally in Section~\ref{subsec:defn}.
Previously,~\cite{GKNR12} gave a  polynomial time $O(\alpha\cdot \log B)$-approximation algorithm for correlated stochastic orienteering. They also showed that this problem is at least as hard to approximate as the deadline orienteering problem, i.e.~an $\Omega(\alpha)$-hardness of approximation (this result also holds for quasi-polynomial time algorithms). Our algorithm improves the approximation ratio to $O(\alpha\cdot \log^2\log B)$, but at the expense of quasi-polynomial running time. We note that the running time in Theorem~\ref{thm:corr-NA} is quasi-polynomial for general inputs where probability distributions are described {\em explicitly}, since the input size is $n\cdot B$. If probability distributions are specified implicitly, the runtime is quasi-polynomial only for $B\le 2^{poly(\log n)}$.

The algorithm in Theorem~\ref{thm:corr-NA} is based on finding an approximate non-adaptive policy, and losing an $O(\log\log B)$-factor on top by Theorem~\ref{thm:corr-loglog-UB}. There are three main steps in the algorithm: (i) we enumerate over $\log B$ many ``portal'' vertices (suitably defined) on the optimal policy; (ii) using these portal vertices, we solve (approximately) a {\em configuration LP relaxation} for paths between portal vertices; (iii) we randomly round the LP solution. The quasi-polynomial running time is only due to the enumeration. In formulating and solving the configuration LP relaxation, we also use some ideas from the earlier $O(\alpha\cdot \log B)$-approximation algorithm~\cite{GKNR12}. Solving the configuration LP requires an algorithm for deadline orienteering (as the dual separation oracle), and incurs an $\alpha$-factor loss in the approximation ratio.
This configuration LP is a ``packing linear program'', for which we can use fast combinatorial algorithms~\cite{PST91,GK07}. The final rounding step involves randomized rounding with alteration, and loses an extra $O(\frac{\log\log B}{\log\log\log B})$ factor.

\subsection{Related Work}
\label{s:prev-work}
The deterministic orienteering problem was introduced by Golden et al.~\cite{GLV87}. It has several applications, and many exact approaches and heuristics have been applied to this problem, see eg. the survey~\cite{VSO11}. The first constant-factor approximation algorithm was due to Blum et al.~\cite{BCKLMM07}. The approximation ratio has been improved~\cite{BBCM04,CKP08} to the current best $2+\epsilon$.

Dean et al.~\cite{DeanGV08} were the first to consider stochastic packing problems in this adaptive optimization framework: they introduced the {\em stochastic knapsack} problem (where items have random sizes), and obtained a constant-factor approximation algorithm and adaptivity gap. The approximation ratio has subsequently been improved to $2+\epsilon$, due to~\cite{BGK11,Bhalgat11}. The stochastic orienteering problem~\cite{GKNR12} is a common generalization of both deterministic orienteering and stochastic knapsack.

Gupta et al.~\cite{GKMR11} studied a generalization of the stochastic knapsack problem, to the setting where the reward and size of each item  may be correlated, and gave an $O(1)$-approximation algorithm and adaptivity gap for this problem.
Recently, Ma~\cite{M14} improved the approximation ratio to $2+\epsilon$.

The correlated stochastic orienteering problem was studied in~\cite{GKNR12}, where the authors obtained an $O(\log n\cdot \log B)$-approximation algorithm and an $O(\log B)$ adaptivity gap. They also showed the  problem to be at least as hard to approximate as the deadline orienteering problem, for which the best approximation ratio known is $O(\log n)$~\cite{BBCM04}.

A related problem to stochastic orienteering was considered by Guha and Munagala~\cite{GuhaM09} in the context of the {\em multi-armed bandit} problem. As observed in~\cite{GKNR12}, the approach in~\cite{GuhaM09} yields an $O(1)$-approximation algorithm (and adaptivity gap) for the variant of stochastic orienteering with two {\em separate} budgets for the travel and processing times. In contrast, our result shows that stochastic orienteering (with a single budget) has super-constant adaptivity gap.

\subsection{Problem Definition}\label{subsec:defn}
%\textbf{Stochastic Orienteering.}
An instance of stochastic orienteering (\so) consists of a metric space
$(V, d)$ with vertex-set $|V| = n$ and symmetric integer distances $d: V \times V \rightarrow \Z^+$ (satisfying the
triangle inequality) that represent travel times.  Each vertex $v \in V$ is associated with a stochastic job, with a deterministic
reward $r_v\ge 0$ and a random processing time (also called size) $\size_v \in \Z^+$ distributed according to a
known probability distribution. % $\pi_v: \R^+ \rightarrow [0,1]$.
The processing times are independent across vertices. We are also given a starting ``root'' vertex $\rho\in V$, and
a budget $B\in \Z^+$ on the total time available. A solution (policy) must start from $\rho$, and visit a sequence of vertices (possibly adaptively). Each job is executed non-preemptively, and the solution knows the precise processing time only upon completion of the job. The objective is to maximize the expected reward from jobs that are completed before the horizon $B$; note that there is no reward for partially completing a job.
The approximation ratio of an
algorithm is the ratio of the expected reward of an optimal policy to that of the algorithm's policy.

We assume that all times (travel and processing) are integer valued and lie in $\{0,1,\cdots,B\}$.
In the {\em correlated} stochastic orienteering problem (\co), the job sizes and rewards are both random, and correlated with each other.
The distributions across different vertices are still independent.  For each vertex $v\in V$, we use $S_v$ and $R_v$ to denote its random size and reward, respectively. We assume an explicit representation of the distribution of each job $v\in V$: for each $s\in\{0,1,\cdots,B\}$, job $v$ has size $S_v=s$ and reward $r_v(s)$ with probability $\Pr[S_v=s]=\pi_v(s)$. Note that the input size is $nB$.

An \emph{adaptive policy} is a decision tree where each node is labeled by a job/vertex of $V$, with the outgoing arcs
from a node labeled by $u$ corresponding to the possible sizes in the support of $S_u$.  A \emph{non-adaptive
policy} is simply given by a path $P$ starting at $\rho$: we just traverse this path, processing
the jobs that we encounter, until the total (random) size of the jobs plus the distance traveled reaches $B$. A
randomized non-adaptive policy may pick a path $P$ at random from some distribution before it knows any of the
size instantiations, and then follows this path as above. Note that in a non-adaptive policy, the order in which jobs
are processed is independent of their processing time instantiations.

%Finally, for any integer $m\ge 0$ we use $[m]$ to denote the set $\{0, 1, \ldots, m\}$.

In our algorithm for \co, we use the {\em deadline orienteering} problem as a subroutine.
The input to this problem is a metric $(U,d)$ denoting travel times, a reward and deadline at each vertex, start ($s$) and end ($t$) vertices, and length bound $D$. The objective is to compute an $s-t$ path of length at most $D$ that maximizes the reward from vertices visited before their deadlines.
The best approximation ratio for this problem is $\alpha=\min\{O(\log n), O(\log B)\}$  due to~\cite{BBCM04,CKP08}.

\subsection{Organization}
The adaptivity gap lower bound appears in Section~\ref{sec:ad-gap}, where we prove Theorem~\ref{thm:ad-gap}. In Section~\ref{sec:corr-ad-gap}, we consider the correlated stochastic orienteering problem and prove the upper bound on its adaptivity gap (Theorem~\ref{thm:corr-loglog-UB}). Finally, the improved quasi-polynomial time algorithm (Theorem~\ref{thm:corr-NA}) for correlated stochastic orienteering appears in Section~\ref{sec:corr-alg}

\section{Lower Bound on the Adaptivity Gap}\label{sec:ad-gap}
Here we describe our lower bound instance which shows that the adaptivity gap is $\Omega(\sqrt{\log \log B})$ even for an undirected line metric. The proof and the description of the instance is divided into three steps. First we describe an instance where the underlying graph is a directed complete binary tree, and prove the lower bound for it. The directedness ensures that all policies follow a path from root to a leaf (possibly with some nodes skipped) without any backtracking. Second, we show that the directed assumption can be removed at the expense of an additional $O(1)$ factor in the adaptivity gap. In particular this means that the nodes on the tree can be visited in any order starting from the root. Finally, we ``embed" the undirected tree into a line metric, and show that the adaptivity gap stays the same up to a constant factor.

\subsection{Directed Binary Tree}

Let $L\ge 2$ be an integer and $p:=1/\sqrt{L}$. We define a complete binary tree \T of height $L$ with root $\rho$.
All the edges are directed from the root towards the leaves.
The {\em level} $\ell(v)$ of any node $v$ is the number of nodes on the shortest path from $v$ to any leaf. So all the leaves are at {\em level} one and the root $\rho$ is at level $L$. We refer to the two children of each internal node as the left and right child, respectively. Each node $v$ of the tree has a job with some deterministic reward $r_v$ and a random size $S_v$. Each random variable $S_v$ is Bernoulli, taking value zero with probability $1-p$ and some positive value $s_v$ with the remaining probability $p$. The budget for the instance is $B = 2^{2^{L+1}}$.

To complete the description of the instance, we need to define the values of the rewards $r_v$, the job sizes $S_v$, and the distances $d(u,v)$ on edges $e=(u,v)$.

{\bf Defining rewards.} For any node $v$, let $\tau(v)$ denote the number of right-branches taken on the path from the root to $v$. We define the reward of each node $v$ to be $r_v:=(1-p)^{\tau(v)}$.

{\bf Defining sizes.} Let $\q{x}:=2^{x}$ for any $x\in \mathbb{R}$. The size at the root, $s_\rho:=\q{2^L}=2^{2^L}$. The rest of the sizes are defined recursively. For any non-root node $v$ at level $\ell(v)$ with $u$ denoting its parent, the size is:
$$s_v:=\left\{
\begin{array}{ll}
s_u\cdot \q{2^{\ell(v)}} & \mbox{ if  $v$ is the right child of $u$}\\
s_u\cdot \q{-2^{\ell(v)}} & \mbox{ if  $v$ is the left child of $u$}
\end{array}\right.
$$

In other words, for a node $v$ at level $\ell$, consider the path $P=(\rho=u_L,u_{L-1},\ldots,u_{\ell+1},u_\ell=v)$  from $\rho$ to $v$. Let $k = \sum_{j=L}^{\ell} (-1)^{i(u_j)} 2^j$
where $i(u_j)=1$ if $u_j$ is the left child of its parent $u_{j+1}$, and $0$ otherwise (we assume $i(\rho)=0$).
Then $s_v = \q{k}$.

Observe that for a node $v$, each node $u$ in its left (resp. right) subtree has $s_u < s_v$ (resp. $s_u > s_v$).

It remains to define distances on the edges. This will be done in an indirect way, and it is instructive to first consider the adaptive policy that we will work with. In particular, the distances will be defined in such a way that the adaptive policy can always continue till it reaches a leaf node.

{\bf Adaptive policy \A.} Consider the policy \A  that goes left at node $u$ whenever it observes size zero at $u$, and goes right otherwise.

Clearly, the {\em residual budget} $b(v)$ at node $v$ under $\A$  will satisfy the following:
$b(\rho) = B = \q{2^{L+1}}= 2^{2^{L+1}}$, and
$$b(v) := \left\{
\begin{array}{ll}
b(u)-s_u - d(u,v) & \mbox{ if  $v$ is the right child of $u$}\\
b(u) - d(u,v) &  \mbox{ if  $v$ is the left child of $u$}
\end{array}\right.
$$

{\bf Defining distances.} We will define the distances so that the residual budgets $b(\cdot)$ under $\A$ satisfy the following:
$b(\rho)=B$, and for any node $v$ with parent $u$,
$$b(v) := \left\{
\begin{array}{ll}
b(u)-s_u & \mbox{ if  $v$ is the right child of $u$}\\
s_u &  \mbox{ if  $v$ is the left child of $u$}
\end{array}\right.
$$

In particular, this implies the following lengths on edges. For any node $v$ with parent $u$,
$$d(u,v) := \left\{
\begin{array}{ll}
0 & \mbox{ if  $v$ is the right child of $u$}\\
b(u)-s_u = b(u)-b(v) &  \mbox{ if  $v$ is the left child of $u$}
\end{array}\right.
$$

In Claim \ref{cl:well-defn} below we will show that the distances are non-negative, and hence well-defined.

%precisely the remaining budget at each node under policy \A. We will show that \A has expected reward $\Omega(\sqrt{L})$ times any non-adaptive policy.
Figure \ref{fig:tree-def} gives a pictorial view of the instance.

\begin{figure}[ht]
  \begin{centering}
    \includegraphics[scale=0.8]{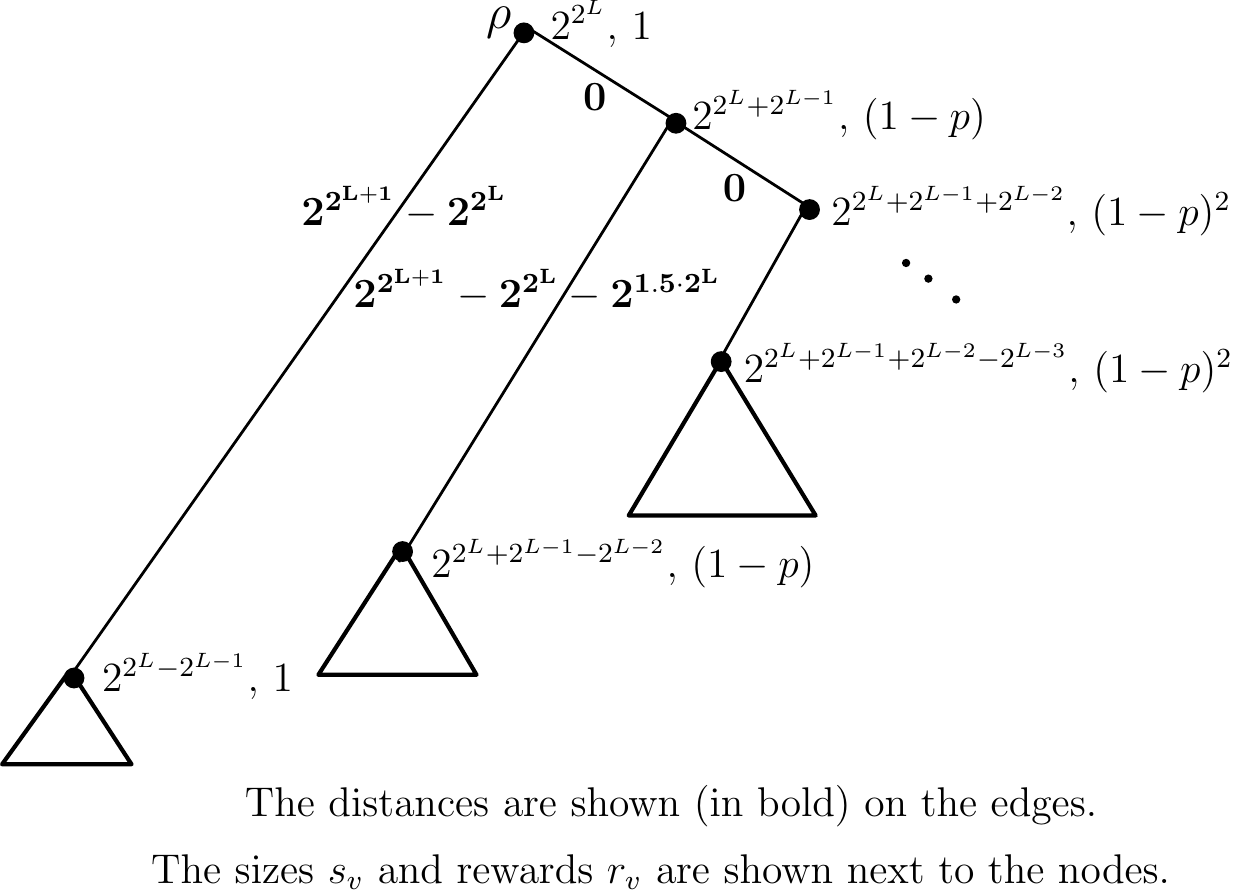}
    \caption{The binary tree \T. \label{fig:tree-def}}
  \end{centering}
\end{figure}

\paragraph{Basic properties of the instance}

Let $A_d(v)$ denote the distance traveled by the adaptive strategy \cal{A} to reach $v$, and let $A_s(v)$ denote the total size instantiation before reaching $v$. By the definition of the budgets, and as \cal{A} takes the right branch at $u$ iff the size at $u$ instantiates, we have the following.

\begin{claim}
\label{lem:bv}
For any node $v$, the budget $b(v)$ satisfies $b(v)=B - A_d(v) - A_s(v)$.
\end{claim}

\begin{claim}
\label{size:val}
If a node $w$ is a left child of its parent, then $b(w) = s_w \cdot \q{2^{\ell(w)}}$.
\end{claim}
\begin{proof}
Let $u$ be the parent of $w$.
By definition of sizes, $s_w = s_u \cdot \q{-2^{\ell(w)}}$.
As $b(w) = s_u$ by the definition of residual budgets, the claim follows.
\end{proof}

\begin{claim}\label{cl:well-defn}
For any node $u$, we have $3\cdot s_u\le b(u)$. This implies that all the residual budgets and distances are non-negative.
 %and that $b(u) \geq b(v)$ for any node $u$ and its child $v$.
\end{claim}
\begin{proof}
Let $w$ denote the lowest level node on the path from $\rho$ to $u$ that is the left child of its parent (if $u$ is the left child of its parent, then $w=u$); if there is no such node, set $w=\rho$.
Note that by Claim~\ref{size:val} and the definition of $s_\rho$ and $b(\rho)$, in either case it holds that $b(w)=s_w\cdot \q{2^{\ell(w)}}$.

Let $\pi$ denote the path from $w$ to $u$ (including $w$ but not $u$; so $\pi=\emptyset$ if $w=u$).
Since $\pi$ contains only right-branches, $b(u)=b(w)-\sum_{y\in \pi} s_y$ and hence $b(u)\ge b(w)- 3 \sum_{y\in \pi} s_y$. Thus to prove $3\cdot s_u\le b(u)$ it  suffices to show $3(s_u+ \sum_{y\in \pi} s_y)\le b(w)$. For brevity, let $s:=s_w$ and $\ell=\ell(w)$. Using the definition of sizes,
\begin{eqnarray*}
s_u+\sum_{y\in \pi} s_y  & \le  & \sum_{i=1}^{\ell} s\cdot \q{2^{\ell-1} + 2^{\ell-2} +\cdots 2^i}  \quad  = \quad   s\cdot \sum_{i=1}^{\ell} \q{2^{\ell}-2^i} \\ & = & s \cdot \sum_{i=1}^{\ell} \q{2^\ell} \cdot \q{-2^i} \quad  \le
\quad s\cdot \q{2^{\ell}} \cdot\sum_{i\ge 1} 4^{-i} \\
 & \le & \frac{1}{3}\cdot s\cdot \q{2^{\ell}} \quad = \quad \frac{b(w)}{3},
\end{eqnarray*}
as desired.  Here the right hand side of the first inequality is simply the total size of nodes in the $w$ to leaf path using all right branches. The inequality in the second line follows as $\q{-2^{i}}= 2^{-2^{i}} \leq 2^{-2i} = 4^{-i}$ for all $i\geq 1$.

Thus we always have $3\cdot s_u\le b(u)$.

As $b(v)=b(u)-s_u$ if $v$ is the right child of $u$, or $b(v) = s_u$ otherwise, this implies that all the residual-budgets are non-negative.

Similarly, as $d(u,v)$ is either $0$ or $b(u)-s_u$ (and hence at least $2/3 b(u)$), this implies that all edge lengths are non-negative.
\end{proof}

This claim shows that the above instance is well defined, and that  \A is a feasible adaptive policy that always continues for $L$ steps until it reaches a leaf. Next, we show that \A obtains large expected reward.

\begin{lemma}\label{lem:ad-profit}
The expected reward of policy \A is $\Omega(L)$.
\end{lemma}
\begin{proof}
Notice that \A accrues reward as follows: it keeps getting reward $1$ (and going left) until the first positive size instantiation, then it goes right for a single step and keeps going left and getting reward $(1-p)$ till the next positive size instantiation and so on. This continues  for a total of $L$ steps.
In particular, at any time $t$ it collects reward $(1-p)^i$, if exactly $i$ nodes have positive sizes among the $t$ nodes seen.

Let $X_i$ denote the Bernoulli random variable that is $1$ if the $i^{th}$ node in \A has a positive size instantiation, and $0$ otherwise. So $E[X_i]=p$, and
$E[X_1 + \ldots + X_L] = L p = \sqrt{L}$. By Markov's inequality, the probability that more than $2\sqrt{L}$
nodes in \A have positive sizes is at most half. Hence, with probability at least $\frac12$ the reward collected in the last node of \A is at least $(1-p)^{2 \sqrt{L}}$. That is, the total expected reward of \A is at least $\frac12\cdot L \cdot (1-p)^{2 \sqrt{L}} \approx L/2 \cdot e^{-2} =  \Omega(L)$.
\end{proof}

\subsection{Bounding Directed Non-adaptive Policies}\label{subsec:dir-gap}
We will first show that any non-adaptive policy \N
that is constrained to visit vertices according to the partial order given by the tree \T
gets reward at most $O(\sqrt{L})$. 
Notice that these correspond precisely to non-adaptive policies on the directed tree \T.

The key property we need from the size construction is the following.
\begin{lemma}\label{lem:prefix-size}
For any node $v$, the total size instantiation observed under the adaptive policy \A before $v$ is strictly less than $s_v$.
\end{lemma}
\begin{proof}
Consider the path $\pi$ from the root to $v$, and let $k_1 < k_2<\cdots<k_t$ denote the levels at which $\pi$ ``turns left''. That is, for each $i$, the node $u_i$ at level $k_i$ in path $\pi$ satisfies (a) $u_i$ is the right child of its parent, and (b) $\pi$ contains the left child of $u_i$ if it goes below level $k_i$. (If $v$ is the right child of its parent then $u_1=v$ and $k_1=\ell(v)$.) Let $s_i$ denote the size of $u_i$, the level $k_i$ node in $\pi$. Also, set $k_{t+1}=L$ corresponding to the root. Below we use $[t]:=\{1,2,\cdots,t\}$.

\begin{figure}[ht]
  \begin{centering}
    \includegraphics[scale=0.7]{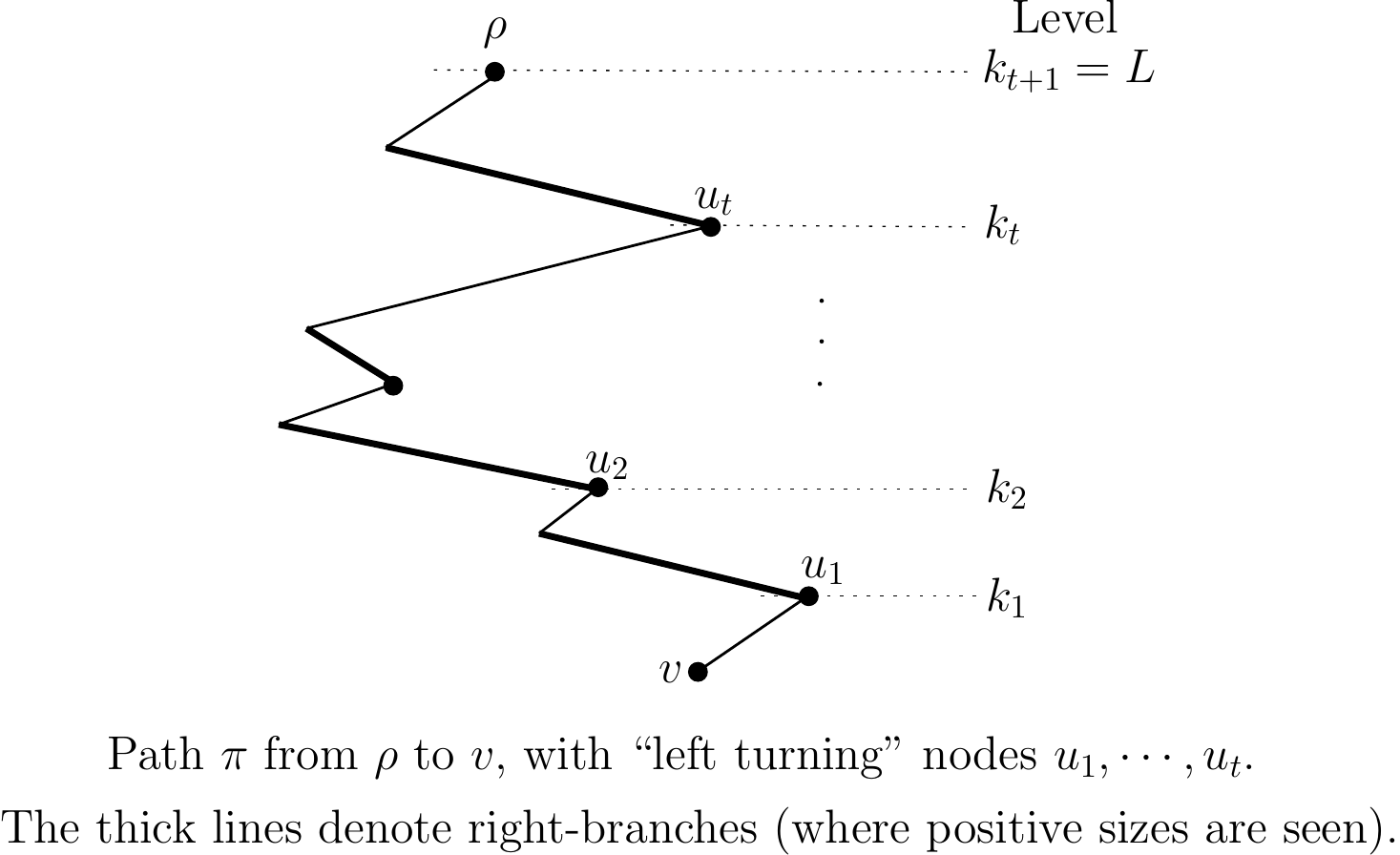}
    \caption{The path $\pi$ in proof of Lemma~\ref{lem:prefix-size}.    \label{fig:??} }
  \end{centering}
\end{figure}

We first bounds the size instantiation between levels $k_{i+1}$ and $k_{i}$ in terms of $s_i$.
Observe that a positive size instantiation is seen in \A only along right branches. So for any $i\in [t]$, the total size instantiation seen in $\pi$ between levels $k_i$ and $k_{i+1}$ %(i.e. the size of nodes in levels $k_i+1,k_i+2,\cdots,k_{i+1}$)
is at most:
\begin{eqnarray}\label{eq:prefix-size1}
& & s_i\cdot \left[ \q{-2^{k_i}} + \q{-2^{k_i}-2^{k_i+1}} + \q{-2^{k_i}-2^{k_i+1} - 2^{k_i+2} }\cdots\right] \nonumber \\
 & \le &  s_i \cdot \q{-2^{k_i}} \cdot \left(1+1/2+1/4+\cdots \right) \quad \leq  \quad 2\,s_i\cdot \q{-2^{k_i}}
\end{eqnarray}

Now, note that for any $i\in[t]$, the sizes $s_{i+1}$ and $s_i$ are related as follows:
\begin{eqnarray}
s_{i+1} &\le & s_i \cdot \q{-2^{k_i} + 2^{k_i+1} + 2^{k_i+2} +\cdots +2^{k_{i+1}-1} }\quad =\quad s_i \cdot \q{-2^{k_i} - 2^{k_i+1} + 2^{k_{i+1}}} \notag \\
&\le &\frac{s_i}{4}\cdot \q{-2^{k_i}  + 2^{k_{i+1}}}  \label{eq:prefix-size2}
\end{eqnarray}
The first inequality uses the fact that the path from $u_{i+1}$ to $u_i$ is a sequence of (at least one) left-branches followed
by a sequence of (at least one) right-branches. As the size decreases along left-branches and increases along right branches, it follows that conditional on the values of $k_{i}$ and $k_{i+1}$, the  ratio $s_{i+1}/s_{i}$ is maximized for the path with a sequence of left branches followed by a single right branch (at level $k_i$).

Using~\eqref{eq:prefix-size2}, we obtain inductively that:
\begin{equation}\label{eq:prefix-size3}
s_{i+1}\cdot \q{-2^{k_{i+1}}} \quad \le \quad \frac14\cdot s_i\cdot \q{-2^{k_i}} \quad \le \quad \frac{1}{4^{i}} \cdot s_1\cdot \q{-2^{k_1}},\qquad \forall i\in[t].
\end{equation}

Using~\eqref{eq:prefix-size1} and~\eqref{eq:prefix-size3}, the total size instantiation seen in $\pi$ (this does not include the size at $v$) is at most:
\begin{equation}\label{eq:prefix-size4}
\sum_{i=1}^t 2\,s_i\cdot \q{-2^{k_i}} \quad \le \quad
2\sum_{i=1}^t \frac{1}{4^{i-1}} \cdot s_1\cdot \q{-2^{k_1}} \quad < \quad 4\,s_1\cdot \q{-2^{k_1}}.
\end{equation}

Finally, observe that the size at the level $k_1$ node $s_1\le s_v\cdot \q{2^{k_1-1}+2^{k_1-2}+\cdots+2^1}=s_v\cdot \q{2^{k_1}-2}$, since $k_1$ is the lowest level at which $\pi$ turns left (i.e. $\pi$ keeps going left below level $k_1$ until $v$). Together with \eqref{eq:prefix-size4}, it follows that the total size instantiation seen before $v$ is strictly less than
$$4 \,s_1\cdot \q{-2^{k_1}} \leq 4 \, \q{-2^{k_1}} \cdot  s_v\cdot \q{2^{k_1}-2}  = 4 \, \q{-2} \, s_v =  s_v.$$
This completes the proof of Lemma~\ref{lem:prefix-size}. \end{proof}

%All size constructions satisfying Claim~\ref{cl:well-defn} and Lemma~\ref{lem:prefix-size} seem to require $B$ to be double exponential in $L$.

We now show that any non-adaptive policy on the directed tree \T achieves reward $O(\sqrt{L})$. Note that any such solution \N is just a root-leaf path in \T that skips some subset of vertices. A node $v$ in \N is an {\em L-branching} node if the path \N goes left after $v$. {\em R-branching} nodes are defined similarly.

\begin{claim}\label{cl:na-Rbranch}
The total reward from R-branching nodes is at most $\sqrt{L}$.
\end{claim}
\begin{proof}
As the reward of a node decreases by a factor of $(1-p)$ upon taking a right branch,
the total reward of such nodes is at most $\sum_{i=0}^L (1-p)^i \le \frac{1}{p}=\sqrt{L}$.
\end{proof}

\begin{claim}\label{cl:na-Lbranch}
\N can not get any reward after two L-branching nodes instantiate to positive sizes.
\end{claim}
\begin{proof}
For any node $v$ in tree \T, let $A_d(v)$ (resp. $A_s(v)$) denote the distance traveled (resp. size instantiated) in the adaptive policy \A until $v$; here $A_s(v)$ does not include the size of $v$.  Observe that Lemma~\ref{lem:prefix-size} implies that $A_s(v)<s_v$ for all nodes $v$.

In the non-adaptive solution \N, let $u$ and $v$ be any two L-branching nodes that instantiate to positive sizes $s_u$ and $s_v$; say $u$ appears before $v$. Under this outcome, we will show that \N exhausts its budget after $v$. Note that the distance traveled to node $v$ in \N is exactly $A_d(v)$, the same as that under \A. So the total distance plus size instantiated in \N is at least $A_d(v)+s_v+s_u$, which (as we show next) is more than the budget $B$.

By Claim  \ref{lem:bv},  $b(v)=B-A_d(v)-A_s(v)$. Moreover, the residual budget $b(u')$ at the left child $u'$ of $u$ equals $s_u$. Since the residual budgets are non-increasing down the tree \T, we have $B-A_d(v)-A_s(v)=b(v) \le b(u') = s_u$, i.e. $A_d(v)\ge B-A_s(v)-s_u$. Hence, the total distance plus size in \N is at least
$$ A_d(v)+s_v+s_u \quad \geq \quad B-A_s(v)+s_v \quad > \quad B,$$
 where the last inequality follows from Lemma~\ref{lem:prefix-size}. So \N can not obtain reward from any node after $v$.
\end{proof}

Combining the above two claims, we obtain:
\begin{claim}\label{cl:mon-na}
The expected reward of any directed non-adaptive policy is at most $3\sqrt{L}$.
\end{claim}
\begin{proof}
Using Claim~\ref{cl:na-Lbranch}, the expected reward from L-branching nodes is at most the expected number of L-branching nodes until two positive sizes instantiate, i.e. at most $\frac{2}{p}=2\sqrt{L}$. Claim~\ref{cl:na-Rbranch} implies that the expected reward from R-branching nodes is at most $\sqrt{L}$. Adding the two types of rewards, we obtain the claim.
\end{proof}

This proves an $\Omega(\sqrt{\log\log B})$ adaptivity gap for stochastic orienteering on directed metrics. We remark that the $O(\log\log B)$ upper bound in~\cite{GKNR12} also holds for directed metrics.

\subsection{Adaptivity Gap for Undirected Tree}

We now show that the adaptivity gap does not change much even if we make the edges of the tree undirected.
In particular, this has the effect of allowing the non-adaptive policy \N to backtrack along the (previously directed) edges, and visit any collection of nodes in the tree.
Recall that in the directed instance of the previous subsection, the non-adaptive policy could not try too many $L$-branching nodes (Claim \ref{cl:na-Lbranch}) and hence was forced to choose mostly $R$-branching nodes, in which case the rewards decreased rapidly.
However, in the undirected case, the non-adaptive policy can move along some right edges to collect rewards and then backtrack to high-reward nodes.

The main idea of the analysis below will be to show that  non-adaptive policies cannot gain much more from backtracking (Claims \ref{cl:na-noLback} and \ref{cl:na-back}).

The adaptive policy we compare against is the same as in the directed case. Let $\N'$ denote some fixed non-adaptive policy.
Using the definition of edge-lengths,
\begin{claim} \label{cl:na-noLback}
$\N'$ can not backtrack over any left-branching edge.
\end{claim}
\begin{proof}
As in the proof of Claim~\ref{cl:na-Lbranch}, for any $v\in\T$, let $A_d(v)$ (resp. $A_s(v)$) denote the distance traveled (resp. size instantiated) in the adaptive policy \A until node $v$; recall $A_s(v)$ does not include the size of $v$. If $\N'$ backtracks over the left-edge $(u,v)$ out of some node $u$ then the distance traveled is at least:
\begin{eqnarray*}
A_d(v) + d(u,v) & = &  A_d(v)+b(u)-b(v) =  B-2\cdot b(v)-A_s(v)+b(u)  \\ & > & \quad B-2\cdot b(v)-s_{v}+b(u)
=  B-2\cdot s_u-s_{v}+b(u) \\
&  \ge &  \quad  B-3\cdot s_u +b(u) \quad\ge \quad B
\end{eqnarray*}
The first equality follows as $d(u,v)=b(u)-b(v)$ and the second equality follows as $A_d(v)=B- A_s(v) - b(v)$ by Claim 
\ref{lem:bv}.
The first inequality follows as $A_s(v)<s_v$,  by Lemma~\ref{lem:prefix-size}. The second inequality follows as $s_v < s_u$ by the definition of sizes. Finally, the last inequality follows by Claim~\ref{cl:well-defn}.
Since the distance traveled by \N' is more than $B$, the claim follows.
\end{proof}

We now focus on bounding the contribution due to backtracking on the right edges.

Let $\{e_i\}_{i=1}^k$ be the left-edges traversed by $\N'$; we denote $e_i=(u_i,v_i)$ where $v_i$ is the left child of $u_i$. We now partition the nodes visited in $\N'$ as follows. For each $i\in[k]:=\{1,2,\cdots k\}$, group $G_i$ consists of nodes visited after traversing $e_{i-1}$ and before traversing $e_i$; and $G_{k+1}$ is the set of nodes visited after $e_k$. Note that the nodes in $G_i$ are visited contiguously using only right edges (they need not be visited in the order given by tree \T, as the algorithm may bactrack). See Figure~\ref{fig:back-NA} for a pictorial view.

For each $i\in[k]$, let $X_i\sse G_i$ denote the nodes at level more than $u_i$ (the parent node of left-edge $e_i$); and  let $Y_i=G_i\setminus X_i\setminus \{u_i\}$. We also set $X_{k+1}=G_{k+1}$.

By using exactly the argument in Claim~\ref{cl:na-Rbranch}, the total reward in $\{X_i\}_{i=1}^{k+1}$ is at most
$\sum_{j=0}^L (1-p)^j \le \sqrt{L}$.

 Let us modify $\N'$ by dropping all nodes in $\{X_i\}_{i=1}^{k+1}$.
 Each remaining node $w$ of $\N'$ is either (i) an L-branching node, where $\N'$ goes left after $w$ (these are the end-points $u_i$s of left-edges), or (ii) a ``{\em backtrack} node'' where $\N'$ backtracks on the edge from $w$ to its parent (these are nodes in $Y_i$s). By Claim~\ref{cl:na-Lbranch}, the expected reward from L-branching nodes is at most $2\sqrt{L}$. In order to bound the total expected reward, it now suffices to bound the reward from the backtrack nodes.

\begin{claim}\label{cl:na-back}
The expected reward of $\N'$ from the backtrack nodes is at most $2\sqrt{L}$
\end{claim}
\begin{proof}
Consider the partition (defined above) of backtrack nodes of $\N'$ into groups $Y_1,Y_2,\cdots,Y_k$. Recall that $\N'$ visits each group $Y_i$ contiguously (perhaps not in the order given by \T) and then traverses left-edge $e_i$ to go to the next group $Y_{i+1}$. Moreover, $u_i$ (the parent end-point of left-edge $e_i$) is an ancestor of all $Y_i$-nodes. See also Figure~\ref{fig:back-NA}.

Note also that the walk visiting each group $Y_i$ consists only of right-edges: so the total reward in any single group $Y_i$ is at most $\min\{|Y_i|, \sqrt{L}\}$ (see Claim~\ref{cl:na-Rbranch}). Define $h_i:=\min\{|Y_i|, \sqrt{L}\}$ for each $i\in[k]$.

\begin{figure}[ht]
  \begin{centering}
    \includegraphics[scale=0.75]{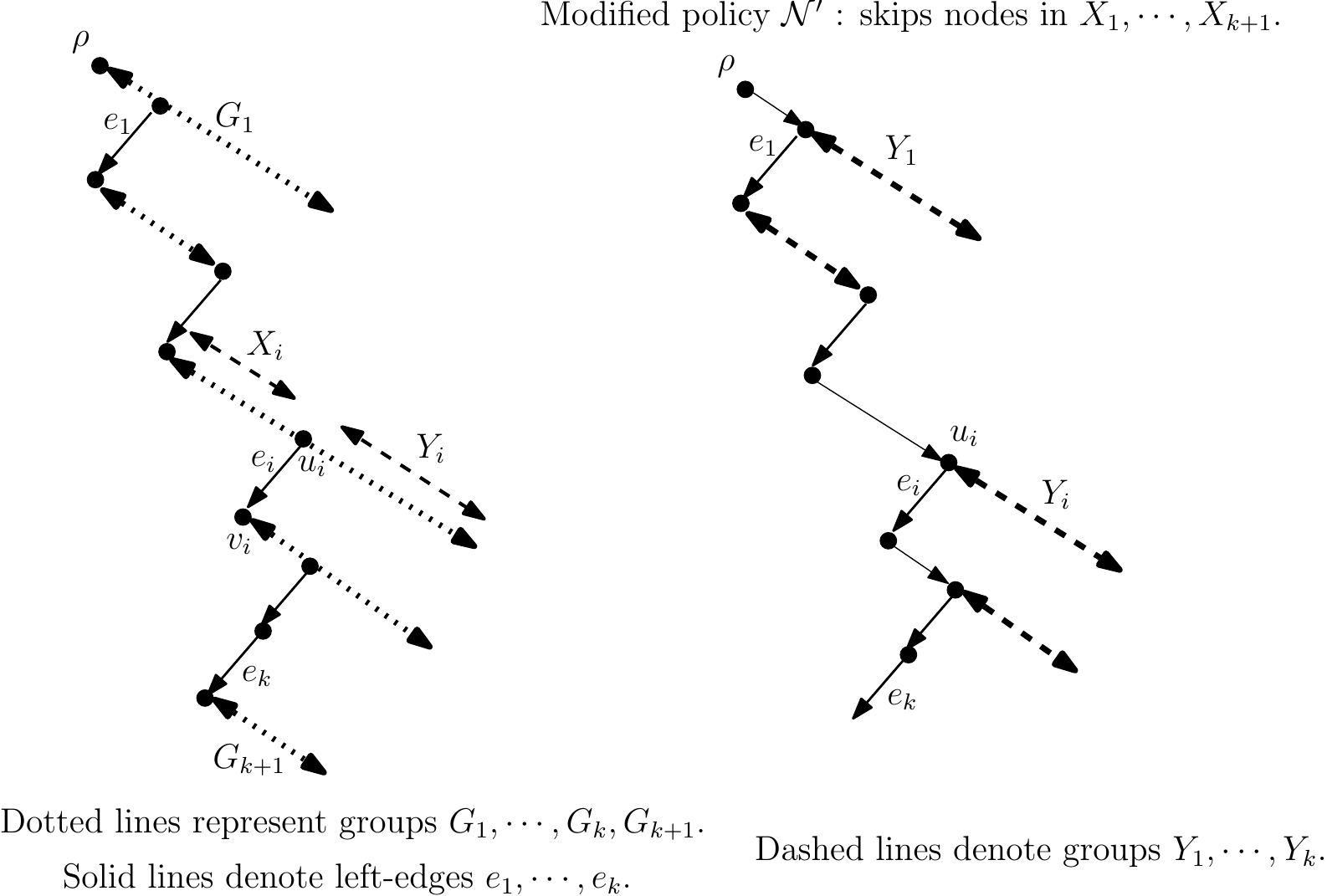}
    \caption{The walk corresponding to non-adaptive policy $\N'$ in Claim~\ref{cl:na-back}. \label{fig:back-NA} }
  \end{centering}
\end{figure}

Let $I$ denote the (random) index of the first group where a positive size is  instantiated. We now show that $\N'$ can not visit any group indexed more than $I$. Let $u=u_I$ and $v=v_I$ denote the end-points of the left-edge $e_I$. Note that $\N'$ must traverse the left edge $e_I=(u,v)$ out of $u$ to reach groups $Y_{I+1},\cdots,Y_k$. If $w\in Y_I$ is the node with positive size instantiation and $\ell$ its level, then $s_w\ge s_u\cdot \q{2^{\ell}}\ge 4\,s_u$ (since $u$ is an ancestor of all $Y_I$-nodes). The distance traveled by \N' till $v$ is $$A_d(v) = B-b(v)-A_s(v) = B - s_u - A_s(v) >B-2\cdot s_u$$
 where the last inequality follows by Lemma~\ref{lem:prefix-size}. Thus the total distance plus size seen in $\N'$ (till $v$) is at least  $B-2s_u+s_w$, which is at least $\ge B+2s_u$ and hence $>B$. Thus $\N'$ can not visit any higher indexed group.

Using the above observation, the expected reward from backtrack nodes is at most:
\begin{eqnarray*}
\E\left[\sum_{i=1}^I h_i\right] & =& \sum_{i=1}^k h_i\cdot \Pr[I\ge i] \quad \le \quad \max_{i=1}^k h_i \,+\, \sum_{i=1}^k h_i\cdot \Pr[I\ge i+1] \\
&\le & \sqrt{L} \,+\, \sum_{i=1}^k h_i\cdot (1-p)^{\sum_{j=1}^{i} |Y_j|} \quad \le \quad \sqrt{L} \,+\, \sum_{i=1}^k h_i\cdot (1-p)^{\sum_{j=1}^{i} h_j} \\
&\le  & \sqrt{L} \,+\, \sum_{x\ge 0} (1-p)^x \quad = \quad 2\sqrt{L}.
\end{eqnarray*}
Above we used the fact that $\Pr[I\ge i+1] = (1-p)^{\sum_{j=1}^i|Y_j|} \le (1-p)^{\sum_{j=1}^i h_j}$.
\end{proof}

Altogether, it follows that any non-adaptive policy $\N'$ has expected reward at most $\sqrt{L} + 2\sqrt{L}+ 2\sqrt{L} =5\sqrt{L}$. Finally, using Lemma~\ref{lem:ad-profit}, we obtain an $\Omega(\sqrt{L})=\Omega\left(\sqrt{\log\log B}\right)$ adaptivity gap.

\subsection{Adaptivity Gap on Line Metric}
We now show that the previous instance on a tree metric can also be embedded into a line metric such that the adaptivity gap does not change much. This gives an $\Omega\left(\sqrt{\log\log B}\right)$ adaptivity gap for stochastic orienteering even on line metrics.

The line metric \lm is defined as follows. Each node $v$ of the tree \T is mapped (on the real line) to the coordinate $d(\rho,v)$ which is the distance in \T from the root $\rho$ to $v$. Since all distances in our construction are integers, each node lies at a non-negative integer coordinate. Note that multiple nodes may be at the same coordinate (for example, as all right-edges in \T have zero length). Below, $d(\cdot)$ will denote distances in the tree metric \T, and $d_\lm(\cdot)$ denotes distances in the line metric \lm.

Note that $d_\lm(\rho,v)=d(\rho,v)$ for all nodes $v\in\T$. Moreover, the distance $d_\lm(u,v)$ between two nodes $u$ and $v$ in the line metric is $|d(\rho,v)-d(\rho,u)|$, which is at most the distance $d_{\T}(u,v)$ in the tree metric. Thus the adaptive policy \A for the tree \T is also valid for the line, which (by Lemma~\ref{lem:ad-profit}) has expected reward $\Omega(L)$.
However, the distances $d_L(u,v)$ on the line could be arbitrarily smaller than $d_T(u,v)$, and thus the key issue is to show that non-adaptive policies cannot do much better. To this end, we begin by observing some more properties of the distances $d(\rho,v)$ and the embedding on the line.

\begin{lemma}
\label{lem:dist-lr}
 For any internal node $u \in \T$, let $L_u$ (resp. $R_u$) denote the subtree rooted at the left (resp. right) child of $u$.
Then, for any node $v \in L_u$, $ d(\rho,v) > B-2 s_u$ and
for any node $v \in R_u$,
$d(\rho,v)  \leq B - 4 s_u.$
\end{lemma}
\begin{proof}
For any node $v$, recall that its residual budget $b(v) = B - d(\rho,v)-A_s(v)$, where $A_s(v)$ is the total size instantiated in the adaptive policy \A before node $v$.
Suppose $v \in L_u$, and let $u'$ be the left child of $u$. Then
 $$d(\rho,v) \geq d(\rho,u') = B - b(u') - A_s(u') = B-s_u - A_s(u') = B - s_u - A_s(u) > B -2 s_u,$$
 where we use that $b(u')=s_u$, $A_s(u')=A_s(u)$ and the last inequality follows from Lemma \ref{lem:prefix-size}.

 Now consider $v \in R_u$.
  We have $A_s(v) \geq s_u$, as $v$ lies in the right subtree under $u$ and so $u$ must have instantiated to a positive size before reaching $v$.
By Claim~\ref{cl:well-defn}, $b(v) \geq 3 s_v$ which is at least $3 s_u$ since $s_v > s_u$ for each $v\in R_u$. Thus $d(\rho,v)=
B - b(v) - A_s(v)  \leq B - 4s_u$.
\end{proof}

This implies the following useful fact.
 \begin{corollary}
 \label{cor:left-right}
 In the line embedding, for any node $u\in \T$, all nodes in the left-subtree $L_u$ appear after all nodes in the right-subtree $R_u$.
 \end{corollary}

We will now show that any non-adaptive policy has reward at most $O(\sqrt{L})$. This requires more work than in the tree metric case, but the high level idea is quite similar:
we restrict how any non-adaptive policy can look like by using the properties of distances, and show that such policies cannot obtain too much profit.
Observe that a non-adaptive policy $\N''$ is just a walk on \lm, originating from $\rho$ and  visiting a subset of vertices.

\begin{lemma}\label{lem:line-na-ordered}
Any non-adaptive policy on \lm must visit vertices ordered by non-decreasing distance from $\rho$.
\end{lemma}
\begin{proof}
We will show that if vertex $v$ is visited before $w$ and $d(\rho,v)>d(\rho,w)$ then the walk to $w$ has length more than $B$; this would prove the lemma.

Let $u$ denote the least common ancestor of $v$ and $w$.
There are two cases depending on whether $u=w$ or $u \notin \{v,w\}$; note that the ancestor $u$ cannot be $v$ as $d(\rho,v) > d(\rho,w)$.

If $u\notin\{v,w\}$, since $d(\rho,v) > d(\rho,w)$, it must be that $v \in L_u$ and $w\in R_u$ by Corollary \ref{cor:left-right}.
Moreover, the total distance traveled by the path is at least
$$d_\lm(\rho,v) + d_\lm(v,w) \ge d(\rho,v) + d(\rho,v) - d(\rho,w)  = 2 d(\rho,v) - d(\rho,w) >  2(B - 2 s_u)  - (B - 4 s_u) = B,$$  where the second inequality is  by Lemma \ref{lem:dist-lr}.

If $u=w$, since $d(\rho,v)>d(\rho,w)$, there must be at least one left edge $e=(x,y)$ on the path from $w$ to $v$ in the tree (as the length of the right edges is 0).  Then,
the distance traveled by the path is at least $d_\lm(\rho,v) + d_\lm(v,w) \ge  d(\rho,y) + d(x,y) = d(\rho,x) + 2d(x,y)$. As $d(\rho,x) = B-b(x) - A_s(x) > B-b(x) - s_x$ by Lemma \ref{lem:prefix-size},
 and as $d(x,y) =  b(x)-s_x$ (by definition of distances on left edges), we have
  $$  d(\rho,x) + 2d(x,y) \geq   (B - b(x) - s_x) + 2(b(x)-s_x) = B + b(x) - 3s_x > B$$
where the last inequality follows from Claim \ref{cl:well-defn}.
\end{proof}

By Lemma~\ref{lem:line-na-ordered}, any non-adaptive policy $\N''$ visits vertices in non-decreasing coordinate order. For vertices at the same coordinate, we can break ties and assume that these nodes are visited in decreasing order of their level in \T. This does not decrease the expected reward due to the following exchange argument.

\def\oN{\ensuremath{\overline{\cal N}}\xspace}

\begin{claim}\label{cl:line-na-ties}
If $\N''$ visits two vertices $\{u,v\}$ consecutively that have the same coordinate in \lm and have levels $\ell(u)>\ell(v)$, then $u$ must be visited before $v$.
\end{claim}
\begin{proof}
Since $u$ and $v$ have the same coordinate in \lm, by Lemma~\ref{lem:dist-lr} it must be that  one is an ancestor of the other, and the $u-v$ path in \T consists only of right-edges. Since $\ell(u)>\ell(v)$, node $u$ is an ancestor of $v$ in \T. Suppose that $\N''$ chooses to visit $v$ before $u$. We will show that the alternate solution \oN that visits $u$ before $v$ has larger expected reward. This is intuitively clear since $u$ stochastically dominates $v$ in our setting: the probabilities are identical, size of $u$ is less than $v$, and reward of $u$ is more than $v$. The formal proof also requires independence of $u$ and $v$, and is by a case analysis.

Let us {\em condition} on all instantiations other than $u$ and $v$: we will show that \oN has larger conditional expected reward than $\N''$. This would also show that the (unconditional) expected reward of \oN is more than $\N''$. Let $X$ denote the total distance plus size in $\N''$ (resp. \oN) when it reaches $v$ (resp. $u$). Irrespective of the outcomes at $u$ and $v$, the residual budgets in $\N''$ and $\oN$ before/after visiting $\{u,v\}$ will be identical. So the only difference in (conditional) expected reward is at $u$ and $v$. The following table lists the different possibilities for  rewards from $u$ and $v$, as $X$ varies (recall that $B$ is the budget).
\begin{center}
{\renewcommand{\arraystretch}{1.5}
\renewcommand{\tabcolsep}{0.2cm}\begin{tabular}{|c|c|c|}
\hline  Cases & Reward $(\N'')$ & Reward $(\oN)$ \\
\hline $X+s_u+s_v\le B$ & $r_u+r_v$ & $r_u+r_v$ \\
\hline $X+s_v\le B<X+s_u+s_v$ & $(1-p^2)r_u+r_v$ & $r_u+(1-p^2)r_v$ \\
\hline $X+s_u\le B<X+s_v$ & $(1-p)(r_u+r_v)$ & $r_u+(1-p)r_v$ \\
\hline $X\le B <  X+s_u$ & $(1-p)^2r_u+(1-p)r_v$ & $(1-p)r_u+(1-p)^2r_v$ \\
\hline $B<X$ & $0$ & $0$ \\
\hline
\end{tabular}}
\end{center}
In each case, $\oN$ gets at least as much reward as $\N''$ since $r_u>r_v$. This completes the proof.
\end{proof}

For any node $v$ in $\N''$, let $E_v$ denote the set of nodes $u$ satisfying (i) $u$ appears before $v$ in $\N''$, and (ii) $u$ is {\em not} an ancestor of $v$ in tree \T. We refer to $E_v$ as the ``blocking set'' for node $v$. We first prove a useful property of the  $\{E_v\}$ sets.
\begin{claim}\label{cl:line-E-prop}
For any $v\in \N''$ and $u\in E_v$, we must have $u\in$ right-subtree$(a)$ and $v\in$ left-subtree$(a)$ at the lowest common ancestor $a$ of $v$ and $u$. Moreover $\N''$ can not get reward from $v$ if any vertex in its blocking set $E_v$ instantiates to a positive size.
\end{claim}
\begin{proof}
Observe that $u$ and $v$ are incomparable in \T because:
\begin{itemize}
\item $u$ is not an ancestor of $v$ by definition of $E_v$.
\item $u$ is not a descendant of $v$. Suppose (for a contradiction) that $u$ is a descendant of $v$. Note that $u$ and $v$ are not co-located in \lm: if it were, then by  Claim~\ref{cl:line-na-ties} and the fact that $\N''$ visits $u$ before $v$, $u$ must be an ancestor of $v$, which contradicts the definition of $E_v$. So the only remaining case is that $u$ is located further from $\rho$ than $v$: but this contradicts Lemma~\ref{lem:line-na-ordered} as $\N''$ visits $u$ before $v$.
\end{itemize}
So the lowest common ancestor $a$ of $v$ and $u$ is distinct from both $v,u$. Since $d(\rho,u)\le d(\rho,v)$, we must have $u\in R_a$ and $v\in L_a$. This proves the first part of the claim.

Since $u\in R_a$, its size $s_u\ge s_a\cdot \q{2^{\ell(a)-1} - 2^{\ell(a)-2}-\cdots -2^1} = 4\cdot s_a$.
As $v \in L_a$, by Lemma \ref{lem:dist-lr} $d(\rho,v)> B-2\cdot s_a$  and hence if $u$ has non-zero size, the total distance plus size until $v$ is more than $4s_a+B-2s_a>B$, i.e.~$\N''$ can not get reward from $v$.
\end{proof}

The next key claim shows that the sets $E_v$ are increasing along $\N''$.
\begin{claim}\label{cl:inc-E}
If node $v$ appears before $w$ in $\N''$ then $E_v\sse E_w$.
\end{claim}
\begin{proof}
Consider nodes $v$ and $w$ as in the claim, and suppose (for contradiction) that there is some $u\in E_v\setminus E_w$. Since $u\in E_v$, by Claim~\ref{cl:line-E-prop}, $u\in$ right-subtree$(a)$ and $v\in$ left-subtree$(a)$, where $a$ is the lowest common ancestor of $u$ and $v$. Clearly $u$ appears before $w$ in $\N''$ ($u$ is before $v$ which is before $w$). And since $u \not\in E_w$, $u$ must be an ancestor of $w$. Hence $w$ is also in the right-subtree$(a)$, and $d(\rho,w)<d(\rho,v)$;  recall that $v\in$ left-subtree$(a)$. This contradicts with Lemma~\ref{lem:line-na-ordered} since $v$ is visited before $w$. Thus $E_v\sse E_w$.
\end{proof}

Based on Claim~\ref{cl:inc-E}, the blocking sets in $\N''$ form an increasing sequence. So we can partition $\N''$ into contiguous segments $\{N_i\}_{i=1}^k$ with $u_i$ (resp. $v_i$) denoting the first (resp. last) vertex of $N_i$,  so that the following hold for each $i\in[k] := \{1,2,\cdots,k\}$.
\begin{itemize}
\item The first vertex $u_i$ of $N_i$ has $|E_{u_i}|\ge (i-1)\sqrt{L}$, and
\item the increase in the blocking set $|E_{v_i}\setminus E_{u_i}| = |E_{v_i}|-|E_{u_i}| < \sqrt{L}$.
\end{itemize}

{\bf Defining directed non-adaptive policies from $\N''$.} For each $i\in[k]$ consider the non-adaptive policy $\N_i$ that traverses segment $N_i$ and visits only vertices in $N_i\setminus E_{v_i}$; note that $N_i\setminus E_{v_i} = N_i\setminus \left(E_{v_i} \setminus E_{u_i}\right)$ since $E_{u_i}\cap N_i=\emptyset$. Notice that the blocking set is always empty in $\N_i$: this means that nodes in $\N_i$ are visited in the order of some root-leaf path in tree \T, i.e. $\N_i$ is a directed non-adaptive policy (as considered in Section~\ref{subsec:dir-gap}). So by Claim~\ref{cl:mon-na}, the expected reward in each $\N_i$ is at most $3\sqrt{L}$. That is,
\begin{equation}\label{eq:line-mon-na}
\max_{i=1}^k \,\, \E\left[\mbox{reward }\N_i\right] \quad \le \quad 3\sqrt{L}
\end{equation}

Now we can upper bound the reward in the original non-adaptive policy $\N''$.
\begin{eqnarray}
\E\left[\mbox{reward }\N''\right] & =& \sum_{i=1}^k \Pr\left[\N'' \mbox{ reaches }u_i\right] \cdot \E\left[\mbox{reward in }N_i \,|\, \N'' \mbox{ reaches }u_i \right] \label{eq:line-na-fin1}\\
&\le & \sum_{i=1}^k e^{-i+1} \cdot \E\left[\mbox{reward in }N_i \, |\, \N'' \mbox{ reaches }u_i \right] \label{eq:line-na-fin2}\\
&\le & \sum_{i=1}^k e^{-i+1} \cdot \left( \E\left[\mbox{reward }\N_i\right] +|E_{v_i}\setminus E_{u_i}|\right) \label{eq:line-na-fin3}\\
&\le & \sum_{i=1}^k e^{-i+1} \cdot 4\sqrt{L}\quad \le \quad \frac{4e}{e-1}\,\sqrt{L}.\label{eq:line-na-fin4}
\end{eqnarray}
%Equality~\eqref{eq:line-na-fin1} is by independence of different sizes.
Inequality~\eqref{eq:line-na-fin2} uses $|E_{u_i}|\ge (i-1)\sqrt{L}$ and the second part of Claim~\ref{cl:line-E-prop} which implies $\Pr[\N''\mbox{ reaches }u_i]\le (1-p)^{|E_{u_i}|}\le e^{-i+1}$. Inequality~\eqref{eq:line-na-fin3} uses the definition of $\N_i$ which is obtained by skipping nodes $E_{v_i}\setminus E_{u_i}$ in $N_i$, and also the independence of the sizes (which allows us to drop the conditioning). Finally,~\eqref{eq:line-na-fin4} uses the property $|E_{v_i}\setminus E_{u_i}|< \sqrt{L}$ and~\eqref{eq:line-mon-na}. This completes the proof of Theorem~\ref{thm:ad-gap}.

%Stochastic
\section{Adaptivity Gap Upper Bound for Correlated  Orienteering} \label{sec:corr-ad-gap}
Here we show that the adaptivity gap of stochastic orienteering is $O(\log\log B)$ even in the setting of {\em correlated} sizes and rewards. This is an extension of the result proved in~\cite{GKNR12} for the uncorrelated case.

The correlated stochastic orienteering problem again consists of a budget $B$ and metric $(V,d)$ with each vertex $v\in V$ denoting a random job. Each vertex has a joint distribution over its reward and processing time (size); so the reward and size at any vertex are correlated. These distributions are still independent across vertices. The basic stochastic orienteering problem is the special case when vertex rewards are deterministic. We prove Theorem~\ref{thm:corr-loglog-UB} in this section.

{\bf Notation.} Recall that for each vertex $v\in V$, $S_v$ and $R_v$ denote its (random) size and reward. As before, all sizes are integers in the range $[B]:=\{0,1,\cdots,B\}$. We assume an explicit representation of each job's distribution: the job at vertex $v\in V$ has size $s\in[B]$ and reward $r_v(s)$ with probability $\Pr[S_v=s]=\pi_v(s)$.

We represent the optimal adaptive policy naturally as a decision tree \T. Nodes in \T are labeled by vertices in $V$ and branches correspond to size (and reward) instantiations. Note that the same vertex of $V$ may appear at multiple nodes of \T (note the distinction between nodes and vertices). However, any root-leaf path in \T contains each vertex at most once. For nodes $u,u'$, we use $u\prec u'$ to denote $u$ being an ancestor of $u'$ in \T, where $u\ne u'$.
For any node $u\in \T$, let $d_u$ (resp. $i_u$) denote the total distance traveled (resp. size observed) in \T before $u$. For node $u\in \T$, we overload notation and use $S_u$, $R_u$ etc to denote the respective term for the vertex labeling $u$.

\def\rr{\ensuremath{\overline{r}}\xspace}

Note that at any node $u\in\T$, only size instantiations $S_u\le B-d_u-i_u$ contribute reward (any larger size violates the budget).
Define the expected reward at node $u$ as $\rr_u := \sum_{s=0}^{B-d_u-i_u} \Pr[S_u=s]\cdot r_u(s)$. Observe that the optimal adaptive reward is $\Opt=\sum_{u\in\T} \Pr[\T\mbox{ reaches }u]\cdot \rr_u$.

For each integer $j=0,1,\cdots,\lceil \log_2B\rceil$ and node $u\in\T$, define $X^j_u:=\min\{S_u,2^j\}$.
%consider an adaptive sequence of random variables

Let $K:=\Theta(\log\log B)$ be a parameter that will be fixed later.

We are now ready to prove Theorem~\ref{thm:corr-loglog-UB}. It is along similar lines as the proof for uncorrelated stochastic orienteering in~\cite{GKNR12}, and also makes use of the following concentration inequality.

\begin{theorem}[Theorem~1 in~\cite{Zhang05}]\label{thm:zhang}
Let $I_1,I_2,\ldots$ be a sequence of possibly dependent random variables; for each $k\ge 1$ variable $I_k$ depends
only on $I_{k-1},\ldots,I_1$. Consider also a sequence of random functionals $\xi_k(I_1,\ldots,I_k)$ that lie in $[0,1]$. Let
$\E_{I_k}[\xi_k(I_1,\ldots,I_k)]$ denote the expectation of $\xi_k$ with respect to $I_k$, conditional on $I_1,\ldots,I_{k-1}$. Furthermore, let
$\tau$ denote any stopping time. Then,
\[ \Pr\left[ \,\,\sum_{k=1}^\tau \E_{I_k}[\xi_k(I_1,\ldots,I_k)] \,\, \ge \,\, \frac{e}{e-1}\cdot \left( \sum_{k=1}^\tau
\xi_k +\delta\right) \right] \quad \le \quad \exp(-\delta), \qquad \forall \delta\ge 0.\]
\end{theorem}

This result is used in proving the following important property.

\begin{lemma}\label{lem:corr-prob}
Assume $K\ge 12$, and fix any $j\in\{0,1,\cdots,\lceil \log B\rceil\}$. Then, the probability of reaching a node $u\in \T$ with:\begin{itemize}
\item $\sum_{v\preceq u} X^j_v \le 2\cdot 2^j$, and
\item $\sum_{v\preceq u} \E_v\left[ X^j_v \right] > K\cdot 2^j$
\end{itemize}
is at most $e^{-K/3}$.
\end{lemma}
\begin{proof}
%This is an application of Theorem~\ref{thm:zhang}.
For each $k=1,2,\cdots$, set $I_k$ to be the $k^{th}$ node seen in \T, and
$$\xi_k(I_1,\ldots,I_k) \quad := \quad \frac{X^j_{I_k}}{2^j} \quad = \quad \min\left\{ \frac{S_{I_k}}{2^{j}}\, , \, 1\right\}.$$
Observe that this sequence satisfies the condition in Theorem~\ref{thm:zhang}, since the identity of the $k^{th}$ node in \T depends only on the outcomes of the previous $k-1$ nodes. Moreover, each $\xi_k(\cdot)$ has a value in the range $[0,1]$. Note that the conditional expectation $\E_{I_k}[\xi_k(I_1,\ldots,I_k)] = \E_{I_k}\left[X^j_{I_k}\right]/2^j$. Define stopping time $\tau$ to be the first node $I_k=u$ (if one exists) at which the following two conditions hold:
$$\sum_{h=1}^k \xi_h(I_1,\ldots,I_h) \le 2, \quad \mbox{ i.e. } \sum_{v\preceq u} X^j_v \le 2\cdot 2^j; \quad \mbox{ and }$$
$$\sum_{h=1}^k \E_{I_h}\left[\xi_h(I_1,\ldots,I_h)\right] > K, \quad \mbox{ i.e. } \sum_{v\preceq u} \E_v\left[ X^j_v \right] > K\cdot 2^j.$$
If there is no such node, then $\tau$ stops when \T ends (i.e. after a leaf-node of \T). Clearly, if $\tau$ stops before \T ends, then
$$ \sum_{k=1}^\tau \E_{I_k}[\xi_k(I_1,\ldots,I_k)] \,\, > \,\, K \,\, > \,\, \frac{e}{e-1}\cdot \left( \sum_{k=1}^\tau
\xi_k(I_1,\ldots,I_k) +\frac{K}{2}-2\right)$$

Now, setting $\delta=\frac{K}{2}-2$ in Theorem~\ref{thm:zhang}, the probability that $\tau$ stops before \T ends (i.e. we reach a node $u$ satisfying the two conditions stated in the lemma) is at most  $e^{-K/2+2}\le e^{-K/3}$  using $K\ge 12$. %This completes the proof.
\end{proof}

\begin{lemma}\label{lem:corr-NAfromA}
Assume $K\ge 3\cdot\log(6\log B)+12$. There is some node $s\in \T$ such that the path $\sigma$ from the root to $s$ satisfies:
\begin{itemize}
\item Total reward: $\sum_{v\in \sigma} \rr_v = \sum_{v\preceq s}\rr_v\ge \Opt/2$.
\item Prefix size: For each $v\in\sigma$ and $j$, either $\sum_{w\preceq v} X^j_w > 2\cdot 2^j$ or $\sum_{w\preceq v} \E_w\left[X^j_w\right] \le K\cdot 2^j$.
\end{itemize}
\end{lemma}
\begin{proof}
For each $j=0,1,\cdots,\lceil \log_2B\rceil$, define {\em band $j$ ``star nodes''} to be those $u\in \T$ that satisfy the two conditions in Lemma~\ref{lem:corr-prob}. Using Lemma~\ref{lem:corr-prob} and a union bound over $1+\lceil\log_2B\rceil$ values of $j$, the probability  of reaching any star node is at most $\frac{3\log_2B}{e^{K/3}}\le \frac12$ since $K\ge 3\cdot\log(6\log B)$.

Observe that for any node $u\in \T$, the conditional expected reward from the subtree of \T under $u$ is at most $\Opt$: otherwise, the alternate policy that visits $u$ directly from the root and follows the subtree below $u$ would be a feasible policy of value more than $\Opt$, contradicting the optimality of \T.

Consider tree \T truncated just before all the star nodes. By the above two properties, the expected reward lost is at most $\Opt\cdot \Pr[\mbox{reach a star node}]\le \Opt/2$. So the remaining reward is at least $\Opt/2$. By averaging, there is some leaf node $s$ in truncated \T, at which $\sum_{v\preceq s} \rr_v\ge \Opt/2$; this proves the ``total reward'' property. Let $\sigma$ denote the path from root to $s$ in \T. Since $\sigma$ does not contain a star node (of any band), every $v\in\sigma$  violates one of the conditions in Lemma~\ref{lem:corr-prob} for each value of $j$. This proves the ``prefix size'' property.
\end{proof}

{\bf The non-adaptive policy.} We now complete the proof of Theorem~\ref{thm:corr-loglog-UB} by implementing the path $\sigma$ from Lemma~\ref{lem:corr-NAfromA} as a non-adaptive solution \N. The policy \N simply involves visiting the vertices on $\sigma$ (in that order) and attempting each job independently with probability $\frac{1}{4K}$. Observe that the distance traveled by \N until any $v\in \sigma$ is exactly $d_v$.

We will show that the expected reward of policy \N is at least $\frac{1}{8K}\sum_{v\in\sigma} \rr_v$.

\begin{claim}\label{cl:corr-NA}
For any $v\in\sigma$, $\Pr\left[\N \mbox{ has seen total size $\le i_v$ until }v\right] \ge \frac34$.
\end{claim}
\begin{proof}
Let $j\in\{0,1,\cdots,\lceil \log_2B\rceil\}$ be the value for which $2^j-1\le i_v<2^{j+1}-1$. So,
$$\sum_{w\prec v} X^j_w \quad =\quad \sum_{w\prec v} \min\{S_w, 2^j\} \quad \le \quad \sum_{w\prec v} S_w \quad =\quad i_v \quad < \quad 2\cdot 2^j.$$

Note that the claim is trivially true for $v$ being the root. For any other $v\in \sigma$, let $v'$ denote the node in $\sigma$ (and \T) immediately preceding $v$.
Using the ``prefix size'' property  (with $v'$ and $j$) in Lemma~\ref{lem:corr-NAfromA}, it follows that $\sum_{w\prec v} \E_w\left[X^j_w\right] \le K\cdot 2^j$. Since policy \N tries each job only with probability $\frac{1}{4K}$,  by Markov's inequality,
$$\Pr\left[\N \mbox{ has seen total size $\ge 2^j$ until }v\right] \le \frac14.$$ Since $i_v\ge 2^j-1$ and sizes are integral, the claim follows.
\end{proof}
We can now bound the reward in \N.
\begin{eqnarray*}
\E[\mbox{reward of }\N] & \ge & \sum_{v\in \sigma} \Pr[\N \mbox{ tries }v]\cdot \Pr[\N \mbox{ has seen total size $\le i_v$ until }v]\cdot \sum_{s=0}^{B-d_v-i_v} \Pr[S_v=s]\cdot r_v(s) \\
&\ge & \sum_{v\in \sigma} \frac{1}{4K}\cdot \frac{3}{4}\cdot \rr_v \quad >\quad \frac{1}{6K}\sum_{v\in \sigma} \rr_v\quad \ge \quad  \frac{\Opt}{12K}.
\end{eqnarray*}
The first inequality uses independence across vertices. The second inequality uses the sampling probability of jobs in \N, Claim~\ref{cl:corr-NA} and the definition of $\rr$. The last inequality uses the ``total reward'' property from Lemma~\ref{lem:corr-NAfromA}. This completes the proof of Theorem~\ref{thm:corr-loglog-UB}.

\paragraph{Defining portals on non-adaptive policy}
We now define some special nodes on the path $\sigma$ from Lemma~\ref{lem:corr-NAfromA}, which will be useful in the approximation algorithm given in the next section. First, some notation.
\begin{definition}[Capped sizes and rewards]\label{def:cap-size-rew}
For any vertex $v$ and integer $j\ge 0$, let $\mu^j_v:=\E[X^j_v]=\E\left[\min\{S_v,2^j\}\right]$ be the mean size of $v$ capped at $2^j$. For any vertex $v$ and integer $d\ge 0$, let $\eta_v(d):=\sum_{s=0}^d \pi_v(s)\cdot r_v(s)$ be the expected reward from $v$ under size instantiation at most $d$.
\end{definition}

Note that the expected reward of any $v\in\sigma$ in \T is $\rr_v=\eta_v(B-d_v-i_v)$.

Also, capped sizes satisfy the following property which will play a crucial role  in our algorithm:
\begin{equation}\label{eq:capped-mu}
\frac{\mu^{j+1}_v}{2^{j+1}} \quad \le \quad \frac{\mu^{j}_v}{2^{j}},\qquad \forall v\in V \mbox{ and } j\ge 0.
\end{equation}
This inequality follows from the fact that $\E\left[\min\{S_v,2^{j+1}\}\right]\le 2\cdot \E\left[\min\{S_v,2^j\}\right]$.

Let $L:=\lceil \log_2B\rceil$, and $[L]:=\{0,1,\cdots,L\}$.  \begin{equation}\label{eq:NA-portals}
\mbox{ Set {\em portal} $v_j$ to be the first vertex $u\in \sigma$ with $i_u\ge 2^{j+1}-1$}, \quad \forall \, j\in [L].
\end{equation}
Recall that $\sigma$ starts at the root $\rho$; for notational convenience set $v_{-1}=\rho$. %let $\rho'$ denote the last vertex on $\sigma$.
Clearly, $v_0\prec v_1\prec\cdots \prec v_{L}$.
For any $j\in [L]$, define {\em segment} $O_j$ to consist of the vertices $v_{j-1}\preceq u\prec v_j$ in path $\sigma$. The following lemma shows that
we can round size instantiations in $\sigma$ to powers of two, and still retain the two properties in Lemma~\ref{lem:corr-NAfromA}.

\begin{lemma}\label{lem:corr-portals}
Given any instance of correlated stochastic orienteering, there exist ``portal vertices'' $\{v_j\}_{j=0}^L$ and path $\sigma$ originating from $\rho$ and visiting the portals in that order, such that:
\begin{itemize}
\item Reward: $\sum_{j\in[L]} \,\, \sum_{u\in O_j} \eta_u(B-d_u-2^j+1) \ge \Opt/2$. For each $j\in[L]$, $O_j$ consists of the vertices in $\sigma$ between $v_{j-1}$ and $v_{j}$. For any $u\in\sigma$, $d_u$ is the distance to vertex $u$ along $\sigma$.
\item Prefix mean size: $\sum_{\ell=0}^j \,\,\sum_{u\in O_\ell} \mu^j_u \le (K+1)\cdot 2^j$, for all $j\in[L]$. Here $K=\Theta(\log\log B)$.
\end{itemize}
\end{lemma}
\begin{proof}
The path $\sigma$ is from Lemma~\ref{lem:corr-NAfromA}, and portals $\{v_j\}_{j=0}^L$ are as in~\eqref{eq:NA-portals}. For the first property, consider any segment $O_j$ and vertex $u\in O_j$. By the definition of portals, we have $i_u\ge i_{v_{j-1}} \ge 2^{j}-1$; so $\rr_u=\eta_u(B-d_u-i_u)\le \eta_u(B-d_u-2^j+1)$. Using the ``total reward'' property in Lemma~\ref{lem:corr-NAfromA}, we have:
$$\sum_{j\in[L]} \,\, \sum_{u\in O_j} \eta_u(B-d_u-2^j+1) \quad \ge \quad \sum_{v\in \sigma} \rr_v \quad \ge \quad \Opt/2.$$
We used the fact that for any vertex $w\in \sigma$ with $v_L\prec w$, $\rr_w=0$ since the total size observed before $w$ is at least $i_{v_L}>B$.

To see the second property consider any $j\in [L]$. Let $w\in \sigma$ be the vertex immediately preceding $v_j$, and let $w'$ be the immediate predecessor of $w$. By definition of portal $v_j$, we have $i_{w}<2^{j+1}$; i.e. $\sum_{u\preceq w'} X^j_u \le i_w < 2^{j+1}$. Using the ``prefix size'' property in Lemma~\ref{lem:corr-NAfromA} with $w'$ and $j$, we obtain that $\sum_{u\prec w}\mu^j_u\le K\cdot 2^j$. So $\sum_{\ell=0}^j \,\,\sum_{u\in O_\ell} \mu^j_u = \sum_{u\prec w}\mu^j_u + \mu^j_w \le (K+1)\cdot 2^j$.
\end{proof}

\section{New Approximation Algorithm for Correlated Orienteering} \label{sec:corr-alg}
In this section, we present an improved quasi-polynomial time approximation algorithm for correlated stochastic orienteering, and prove Theorem~\ref{thm:corr-NA}.

An important subroutine in our algorithm is the {\em deadline orienteering} problem~\cite{BBCM04}. The input to deadline-orienteering is a metric $(U,d)$ denoting travel times, rewards $\{r_v\}_{v\in U}$ and deadlines $\{\Delta_v\}_{v\in U}$ at all vertices, start ($s$) and end ($t$) vertices, and length bound $D$. The objective is to compute an $s-t$ path of length at most $D$ that maximizes the reward from vertices visited before their deadlines. The best approximation ratio for this problem is $\alpha=\min\{O(\log n), O(\log B)\}$  due to Bansal et al.~\cite{BBCM04}; see also Chekuri et al.~\cite{CKP08}. (Strictly speaking, this definition is slightly more general than the usual one where there is no end-vertex or length bound; but all known approximation algorithms also work for the version we use here.) We actually need an algorithm for a generalization of this problem, in the presence of an additional knapsack constraint. The input to {\em knapsack deadline orienteering} (\kdo) is the same as deadline-orienteering, along with a knapsack constraint given by sizes $\{a_v:v\in U\}$ and capacity $A$. The objective is an  $s-t$ path of length at most $D$, having total knapsack size at most $A$, that maximizes the reward from vertices visited before their deadlines. We will use the following known result for \kdo.
\begin{theorem}[\cite{GKNR12}] \label{thm:kdo} There is an $O(\alpha)$-approximation algorithm for knapsack deadline orienteering, where $\alpha$ denotes the best approximation ratio for the deadline orienteering problem.
\end{theorem}

Previously, a polynomial time $O(\alpha\cdot \log B)$-approximation algorithm and $\Omega(\alpha)$-hardness of approximation were known for correlated stochastic orienteering~\cite{GKNR12}, where $\alpha$ is the best approximation ratio for deadline orienteering. Our result improves the approximation ratio to $O(\alpha\cdot \log^2\log B)$, at the expense of quasi-polynomial running time.

{\bf Outline:} The algorithm involves three main steps. First, it guesses (by enumeration) $\log B$ many ``portal vertices'' corresponding to the near-optimal structure in Lemma~\ref{lem:corr-portals}, as well as the distances traveled between consecutive portal vertices. (This is the only step that requires quasi-polynomial time.) Next, based on this information, the algorithm solves a configuration LP-relaxation for paths between the portal vertices. This step also makes use of some results/ideas from the previous algorithm~\cite{GKNR12}. Finally, the algorithm uses randomized rounding with alterations to compute the non-adaptive policy from the LP solution.

\paragraph{Portals.}
The simple enumeration algorithm guesses the $L+1$ portal vertices $\{v_j\}_{j=0}^L$ as in Lemma~\ref{lem:corr-portals}; recall that $L=\lceil \log_2B\rceil$. It also guesses the lengths $\{D_j\}_{j=0}^L$ of the segments $O_j$s in Lemma~\ref{lem:corr-portals}. This requires running time $(nB)^{L+1}$.

We can reduce the enumeration of lengths using a somewhat stronger property than in Lemma~\ref{lem:corr-portals}.

\begin{lemma}\label{lem:corr-enum}
Given any instance of correlated stochastic orienteering, there exist portal vertices $\{v_j\}_{j=0}^L$, auxiliary vertices $\{m_j\}_{j=0}^L$, integers $\{e_j \in [L]\}_{j=0}^L$, and for each $j\in [L]$ a $v_{j-1}$-$v_j$ path $P_j$, such that:
\begin{itemize}
\item Path length: for any $j\in[L]$, path $P_j$ has  length at most $D_j := d(v_{j-1},m_j)+d(m_j,v_j)+2^{e_j}-1$.
\item Reward: $\sum_{j\in[L]} \,\, \sum_{u\in P_j} \eta_u(B-\sum_{i=0}^{j-1} D_i - t_u - 2^j+1) \ge \Opt/4$. Here, $t_u$ is the distance to vertex $u\in P_j$ from $v_{j-1}$ along path $P_j$.
\item Prefix mean size: $\sum_{\ell=0}^j \,\,\sum_{u\in P_\ell} \mu^j_u \le (K+1)\cdot 2^j$, for all $j\in[L]$. Here $K=\Theta(\log\log B)$.
\end{itemize}
\end{lemma}
\begin{proof}
Consider the path $\sigma$ and portal vertices $\{v_j\}_{j=0}^L$ satisfying the properties in Lemma~\ref{lem:corr-portals}; recall that $v_{-1}=\rho$. For each $j\in[L]$, let $O_j$ denote the subpath of $\sigma$ from vertex $v_{j-1}$ to $v_j$. Recall that for any $u\in \sigma$, $d_u=$ distance from $\rho$ to $u$ along path $\sigma$.

For each $j\in[L]$, let $D'_j$ denote the length of subpath $O_j$. Also define for each $j\in[L]$ and vertex $u\in O_j$, $t_u=$ distance to vertex $u$ from $v_{j-1}$ along path $O_j$; note that $d_u=\sum_{i=0}^{j-1} D'_i + t_u$.

We now modify each subpath $O_j$ to obtain a new $v_{j-1}$-$v_j$ subpath $P_j$ as follows. For each vertex $u\in O_j$ define its profit $q_u:=\eta_u(B-d_u-2^j+1)$, and let $q(O_j)$ denote the total profit of vertices in $O_j\setminus \{v_{j-1}\}$. Let $m_j$ denote any vertex on $O_j$ such that:
\begin{equation} \label{eq:portal-mid} \sum_{v_{j-1}\prec u \preceq m_j} \, q_u \,\,\ge \,\,\frac{q(O_j)}{2}\quad \mbox{and} \quad  \sum_{m_{j}\preceq u \preceq v_j} \, q_u \,\, \ge \,\,\frac{q(O_j)}{2} .
\end{equation}
Here $\prec$ denotes the order in which vertices appear on $O_j$. So $m_j$ can be viewed as a ``mid point'' of $O_j$, containing half the total profit on either side of it.

Let $\ell_1=$ length of $O_j$ from $v_{j-1}$ to $m_j$, and $b_1=d(v_{j-1},m_j)$. Similarly, let $\ell_2=$ length of $O_j$ from $m_{j}$ to $v_j$, and $b_2=d(m_{j},v_j)$. Also, let $\epsilon_1=\ell_1-b_1$, $\epsilon_2=\ell_2-b_2$ and $\epsilon=\epsilon_1+\epsilon_2$; clearly $\epsilon_1,\epsilon_2\ge 0$. Note that the length of $O_j$ is $D'_j=b_1+b_2+\epsilon$. Let $O'_j$ denote the subpath  $v_{j-1} \rightsquigarrow m_j \rightarrow v_j$ of $O_j$ that shortcuts over vertices between $m_j$ and $v_j$. Similarly, $O''_j$ denotes the subpath  $v_{j-1} \rightarrow m_j \rightsquigarrow v_j$ of $O_j$ that shortcuts over vertices between $v_{j-1}$ and $m_j$. By~\eqref{eq:portal-mid} the total profit on each of $O'_j$ and $O''_j$ is at least $q(O_j)/2$.
The length of subpath $O'_j$ (resp. $O''_j$) is $\ell_1+b_2$ (resp. $b_1+\ell_2$). We have:
$$\min\{\ell_1+b_2, b_1+\ell_2\} \,\, = \,\,  \min\{b_1+b_2+\epsilon_1, b_1+b_2 + \epsilon_2\}   \,\, \le  \,\, b_1+b_2 +\frac{\epsilon_1+\epsilon_2}{2}  \,\, =  \,\, b_1+b_2 +\frac{\epsilon}2.$$

Let $e_j$ denote the unique integer such that $2^{e_j}-1\le \epsilon< 2^{e_j+1}-1$; note that $e_j\in [L]$ since $\epsilon$ is an integer between $0$ and $B$. Based on the above inequality, the shorter of $O'_j$ and $O''_j$ has length at most $b_1+b_2+2^{e_j}-1$. We set subpath $P_j$ to be the shorter one among $O'_j$ and $O''_j$. This choice ensures that:
\begin{equation} \label{eq:enum-subpath}
\mbox{ length of } P_j \,\le\, d(v_{j-1},m_j)+d(m_j,v_j)+2^{e_j}-1=:D_j,\,\,\mbox{and}\,\, \sum_{u\in P_j} \eta_u(B-d_u-2^j+1)\,\ge \, \frac{q(O_j)}{2}.
\end{equation}
Note also that $D'_j=b_1+b_2+\epsilon\ge b_1+b_2+2^{e_j}-1=D_j$.

We set the portal vertices $\{v_j\}_{j=0}^L$, auxiliary vertices $\{m_j\}_{j=0}^L$, integers $\{e_j \in [L]\}_{j=0}^L$, and path $\{P_j\}_{j=0}^L$ as defined above.  The path length property follows from the first condition in~\eqref{eq:enum-subpath} and the definition of $D_j$. The prefix size property is immediate from that in Lemma~\ref{lem:corr-portals} since each $P_j\sse O_j$.
We now show the reward property:
\begin{eqnarray}
\sum_{j=0}^L \,\, \sum_{u\in P_j} \eta_u(B-\sum_{i=0}^{j-1} D_i - t_u - 2^j+1) &\ge & \sum_{j=0}^L \,\, \sum_{u\in P_j} \eta_u(B-\sum_{i=0}^{j-1} D'_i - t_u - 2^j+1) \label{eq:enum-profit1}\\
&=& \sum_{j=0}^L \,\, \sum_{u\in P_j} \eta_u(B- d_u - 2^j+1) \label{eq:enum-profit2}\\
&=& \sum_{j=0}^L \,\, \sum_{u\in P_j} q_u \quad \ge \quad \frac{1}{2} \sum_{j=0}^L \,\, \sum_{u\in O_j} q_u \label{eq:enum-profit3}\\
&= & \frac{1}{2} \sum_{j=0}^L \sum_{u\in O_j} \eta_u(B- d_u - 2^j+1)\quad \ge \quad \frac{\Opt}{4}. \label{eq:enum-profit4}
\end{eqnarray}
Inequality~\eqref{eq:enum-profit1} uses the fact that $D_i\le D'_i$ for all $i\in[L]$. The  equality~\eqref{eq:enum-profit2} uses $d_u=\sum_{i=0}^{j-1} D'_i + t_u$ for any $u\in P_j\sse O_j$ and $j\in[L]$. Inequality~\eqref{eq:enum-profit3} is by the definition of $q_u$ and the second condition from~\eqref{eq:enum-subpath}. Inequality~\eqref{eq:enum-profit4} uses the reward property in Lemma~\ref{lem:corr-portals}.
\end{proof}
Based on Lemma~\ref{lem:corr-enum}, our enumeration algorithm guesses the $L+1$ portals $\{v_j\}_{j=0}^L$ and segment lengths $\{D_j\}_{j=0}^L$ in time $n^{2L+2}\cdot L^{L+1}=(n\log B)^{O(\log B)}$ (by 
\eqref{eq:enum-subpath} there are at most $nL$ choices for $D_j$, as it is determined by  $e_j  \in [L]$ and $m_j 
\in V$).

\paragraph{Knapsack Deadline Orienteering Instances.} Our goal is to find a $v_{j-1}-v_j$ path $P_j$ for each $j\in[L]$, that have properties similar to the segments in Lemma~\ref{lem:corr-enum}. To this end, we define a knapsack-deadline-orienteering instance $\I_j$ for each $j\in[L]$.

Instance $\I_j$ is defined on metric $(V,d)$ with start vertex $v_{j-1}$, end vertex $v_j$ and length bound $D_j$. The rewards, sizes and deadlines will be defined shortly. As in~\cite{GKNR12}, we will introduce a suitable set of {\em copies} of each vertex $u\in V$, with deadlines that correspond to visiting $u$ at different possible times.
\begin{equation}\label{eq:kdo-copies}
\mbox{ For each $v\in V$ and distance $d\in [B]$, define $f_{v}^j(d):=\eta_v(B-\sum_{i=0}^{j-1}D_i-d-2^j+1)$, }
\end{equation} the expected reward from $v$ under instantiation at most $B-\sum_{i=0}^{j-1}D_i-d-2^j+1$. Intuitively, $f^j_v(d)$ is the expected reward obtained by visiting vertex $v$ in segment $j$ at distance $d$ along $P_j$ (so the total distance to $v$ is $\sum_{i=0}^{j-1}D_i+d$) and having observed total size $2^j-1$ until $v$.

In the \kdo instance $\I_j$, we would like to introduce a copy of vertex $v$ for each distance $d\in [B]$, in order to permit all possibilities for visiting $v$ in segment $j$.
However, solutions to such an instance may obtain a large reward just by visiting multiple copies of the same vertex. In order to control the total reward obtainable from multiple copies of a vertex, we introduce copies of each vertex corresponding only to a suitable subset of $[B]$. This relies on the following construction.
\begin{claim}[\cite{GKNR12}] \label{cl:gknr-count}
For each $j\in[L]$ and $v\in V$, we can compute in polynomial time, a subset $I^j_v\sse [B]$ s.t.
$$\frac{1}{3}\cdot \sum_{y\in I^j_v:y\ge d} f^j_v(y) \quad \le \quad f^j_v(d) \quad \le \quad 2\cdot \max_{y\in I^j_v:y\ge d} f^j_v(y),\qquad \forall d\in[B].$$
\end{claim}
%\vspace{-3mm}
Roughly speaking, the copies $I^j_v$ of vertex $v$ are those times $t$ at which its expected reward $f^j_v(t)$ doubles.

We now formally define the entire instance $\I_j$.
\begin{definition}[\kdo instance $\I_j$]
The metric is $(V,d)$, start-vertex $v_{j-1}$, end-vertex $v_j$, length bound $D_j$ and knapsack capacity $(K+1)\cdot 2^j$. For each $v\in V$ and $t\in I^j_v$, a job $\langle v, j, t\rangle$ with reward $f^j_v(t)$, deadline $t$ and knapsack-size $\mu^j_v$ is located at vertex $v$.
\end{definition}

To reduce notation, when it is clear from context, we will use $u,v$ etc. to refer to jobs as well. A feasible solution $\tau$ to $\I_j$ is a $v_{j-1}-v_j$ path of length $d(\tau)\le D_j$ and total size $\sum_{u\in \tau} \mu^j_u \le (K+1)\cdot 2^j$; the objective value is the total reward of jobs visited by $\tau$ before their deadlines. We will also use $\I_j$ to denote the set of all feasible solutions to this \kdo instance.

{\bf Remark:} A solution $\tau$ to this instance $\I_j$ corresponds to the $j^{th}$ segment in the correlated orienteering solution (between portal $v_{j-1}$ and $v_j$). Note that $\I_j$ imposes no upper bound on
$\tau$'s mean size capped at $2^h$ (i.e. $\mu^h(\tau)$) for any $h\ne j$, which however will be required in the the expected reward analysis (this corresponds to the ``prefix mean size'' property
in Lemma~\ref{lem:corr-enum}). Although there is no explicit bound on $\mu^h(\tau)$ for $h\ne j$, we can infer this using the capped mean properties, as follows.
\begin{itemize}
\item For any $h<j$, we do not need a bound on $\mu^h(\tau)$ since $\tau$ does not affect segment $h$.
\item For any  $h>j$, we can use~\eqref{eq:capped-mu} to obtain
$$\frac{\sum_{v\in \tau} \mu^{h}_v}{2^{h}} \quad \le \quad \frac{\sum_{v\in \tau} \mu^{j}_v}{2^{j}}\quad \le \quad K+1.$$
So the knapsack constraint in $\I_j$ corresponding to $\mu^j(\tau)$ also implicitly bounds $\mu^h(\tau)$.
\end{itemize}
The fact that we introduce only a {\em single} knapsack constraint in $\I_j$ turns out to be important since the best algorithm (that we are aware of) for deadline-orienteering with multiple knapsack constraints has an approximation factor that grows linearly with the number of knapsacks.

\paragraph{Configuration LP relaxation.} We now give an LP
relaxation to the near-optimal structure in Lemma~\ref{lem:corr-enum}. We make use of the guessed portals $\{v_j\}_{j=0}^L$ and lengths $\{D_j\}_{j=0}^L$. The segments $P_j$s in Lemma~\ref{lem:corr-enum} will correspond to solutions of \kdo instances $\I_j$s. The LP relaxation is given below:

\begin{alignat}{2}
  \mbox{max } \,\, \sum_{j=0}^L \, \sum_{v \in V}  \, \sum_{t \in I^j_v} &  \, f^j_v(t)\cdot y^j_{v,t}  & & \label{LP:0} \tag{\lp}  \\
  \mbox{s.t. } \,\, \sum_{\tau \in \I_j} x^j_{\tau} &  \, \le  \, 1 & \qquad & \forall \, j \in [L]   \label{LP:1} \\
  y^j_{v,t} & \, \leq  \,  \, \sum_{\tau\ni \langle v,j,t\rangle} x^j_\tau & \qquad & \forall \, j\in [L],\, v \in
  V, \, t\in I^j_v \label{LP:2} \\
  \sum_{j=0}^L  \, \sum_{t\in I^j_v} y^j_{v,t} & \, \leq  \, 1 & \qquad & \forall \, v \in  V \label{LP:3} \\
  \sum_{v \in V} \mu^j_v\cdot \sum_{\ell=0}^j  \, \sum_{t \in I^\ell_v}  y^\ell_{v,t} & \, \leq  \, (K+1)\cdot 2^j & \qquad & \forall \, j \in [L] \label{LP:4} \\
  \mathbf{x},\, \mathbf{y} &\geq 0. & & \label{LP:5}
\end{alignat}
Above, $\{x^j_\tau: \tau\in \I_j\}$ is intended to be the indicator variable that chooses a solution path for the \kdo instance $\I_j$. Variables $y^j_{v,t}$ indicate whether/not job $\langle v,j,t\rangle$ is selected in $\I_j$. These are enforced by constraints~\eqref{LP:1} and~\eqref{LP:2}. Constraint~\eqref{LP:3} requires at most one copy of each vertex to be chosen, over all segments. And constraint~\eqref{LP:4} bounds the ``prefix mean size'' as in Lemma~\ref{lem:corr-enum}. The objective corresponds to the reward in  Lemma~\ref{lem:corr-enum}.

First, observe that we can ensure equality in~\eqref{LP:2}.
\begin{claim}\label{cl:xy-equality}
Any \lp solution $(x,y)$ can be modified to another solution that satisfies~\eqref{LP:2} with equality and has the same objective value.
\end{claim}
\begin{proof} This follows immediately from the fact that each solution-set $\I_j$ is ``down monotone'', i.e. $I\in\I_j$ and $I'\sse I$ implies $I'\in \I_j$. Formally, consider the constraints~\eqref{LP:2} in any order. For each $\langle v,j,t\rangle$, order arbitrarily the $x^j_\tau$ variables with $\tau \ni \langle v,j,t\rangle$ as $x^j_{\tau(1)}, x^j_{\tau(2)},\cdots,x^j_{\tau(a)}$. Let $b$ denote the unique index with $\sum_{k=1}^{b-1} x^j_{\tau(k)} < y^j_{v,t} \le \sum_{k=1}^{b} x^j_{\tau(k)}$. Now perform the following modification:
\begin{itemize}
\item For each index $k>b$, drop $\langle v,j,t\rangle$ from the solution $\tau(k)$.
\item For index $b$, consider the solutions $\sigma=\tau(b)$ and $\sigma'=\tau(b)\setminus \langle v,j,t\rangle$. Set $x^j_{\sigma'}= \sum_{k=1}^{b} x^j_{\tau(k)} - y^j_{v,t}$ and $x^j_{\sigma} = y^j_{v,t} -  \sum_{k=1}^{b-1} x^j_{\tau(k)}$.
\end{itemize}
After this change, it is clear that we have equality for
$\langle v,j,t\rangle$ in constraint~\eqref{LP:2}. Also, constraint~\eqref{LP:1} remains feasible. Finally, since the $y$ variables remain unchanged, constraints~\eqref{LP:3}-\eqref{LP:4} remain feasible and the objective value stays the same.
\end{proof}

\begin{claim}\label{cl:LP-obj}
The optimal value of \lp is at least $\Opt/8$.
\end{claim}
\begin{proof}
We will show that the subpaths $\{P_j\}_{j=0}^L$ in Lemma~\ref{lem:corr-enum} correspond to a feasible integral solution to \lp. 
Let $\sigma$ denote the concatenated path $P_0,P_1,\cdots,P_L$.
Based on our guess of the portals and distances, we have $d(P_j)\le D_j$. For any vertex $u\in P_j$ let $t_u$ denote the distance to $u$ along $P_j$ (this also appears in Lemma~\ref{lem:corr-enum}); note that the distance to $u$ along $\sigma$ is then $d_u=\sum_{i=0}^{j-1}D_i + t_u$. For each $u\in P_j$ let $t'_u:=\min\{t\in I^j_u : t\ge t_u\}$ be the deadline of the earliest copy of $u$ that path $P_j$ can visit in the \kdo instance $\I_j$. Consider $P_j$ as a solution to $\I_j$ which visits jobs $\{\langle u,j,t'_u\rangle : u\in P_j\}$; in order to verify feasibility, we only need to check that $\sum_{u\in P_j} \mu^j_u\le (K+1)\cdot 2^j$, which follows from the third property in Lemma~\ref{lem:corr-enum}.

Consider now the solution to \lp with $x^j_{P_j}=1$ for all $j\in [L]$; $y^j_{u,t'_u}=1$ for each $u\in P_j$, $j\in[L]$; and all other variables set to zero. Constraints~\eqref{LP:1},~\eqref{LP:2},~\eqref{LP:3} and ~\eqref{LP:5} are clearly satisfied. Observe that the left-hand-side of~\eqref{LP:4} is exactly $\sum_{\ell=0}^j \sum_{u\in P_\ell} \mu^j_u$ which is at most $(K+1)\cdot 2^j$ by the third property in Lemma~\ref{lem:corr-enum}. So $(\mathbf{x}, \mathbf{y})$ is a feasible (integral) solution to \lp. We now bound the objective value:
\begin{eqnarray*}
\sum_{j=0}^L \sum_{u\in P_j} f^j_u(t'_u) &\ge &\frac12 \cdot \sum_{j=0}^L \sum_{u\in P_j} f^j_u(t_u) \\
&= & \frac12 \cdot \sum_{j=0}^L \sum_{u\in P_j} \eta_u\big(B-\sum_{i=0}^{j-1}D_i -t_u-2^j+1\big) \\
&= & \frac12 \cdot \sum_{j=0}^L \sum_{u\in P_j} \eta_u\left(B-d_u-2^j+1\right)\quad \ge \quad \frac{\Opt}{8}.
\end{eqnarray*}
The first inequality is by definition of $t'_u$ and Claim~\ref{cl:gknr-count}; the next two equalities use the definitions $f^j_u(\cdot)$ and $d_u$; and the final inequality is by the first property in Lemma~\ref{lem:corr-enum}.
\end{proof}

\paragraph{Solving the configuration LP.} We show that \lp  can be solved approximately using an approximation algorithm for \kdo. This is based on applying the Ellipsoid algorithm to the dual of \lp, which is given below:
\begin{alignat}{2}
  \mbox{min } \,\, \sum_{j=0}^L \, \beta_j \,\,+\,\, \sum_{v \in V}  \,\delta_v \,\, +  & \,\, (K+1)\cdot \sum_{j=0}^L \, 2^j\cdot z_j & & \label{DLP:0} \tag{\dlp}  \\
  \mbox{s.t. } \,\, \gamma^j_{v,t} + \delta_v + \sum_{\ell\ge j} \mu^\ell_v\cdot z_\ell &  \, \ge  \, f^j_v(t) & \qquad & \forall \, j \in [L], \, v\in V,\, t\in I^j_v   \label{DLP:1} \\
  -\sum_{\langle v,j,t\rangle\in \tau} \gamma^j_{v,t} \,+\, \beta_j & \, \geq  \,  \, 0 & \qquad & \forall \, j\in [L],\, \tau \in
\I_j \label{DLP:2} \\
  \mathbf{\beta},\,\mathbf{\gamma},\, \mathbf{\delta},\, z &\geq 0. & & \label{DLP:3}
\end{alignat}
In order to solve \dlp using the Ellipsoid method, we need to provide a separation oracle that tests feasibility. Observe that constraints~\eqref{DLP:1} are polynomial in number, and can be checked explicitly. Constraint~\eqref{DLP:2} for any $j\in[L]$ is equivalent to asking whether the optimal value of \kdo instance $\I_j$ with rewards $\{\gamma^j_{v,t} : v\in V,\, t\in I^j_v\}$ is at most $\beta_j$. Using an $O(\alpha)$-approximate separation oracle (the knapsack deadline orienteering algorithm from Theorem~\ref{thm:kdo}) within the Ellipsoid algorithm, we obtain an $O(\alpha)$-approximation algorithm for \dlp and hence \lp.

{\bf Remark:} Alternatively, we can solve \lp using a faster combinatorial algorithm that is based on multiplicative weight updates. Note that we can eliminate $y$-variables in \lp by setting constraints~\eqref{LP:2} to equality (Claim~\ref{cl:xy-equality}). This results in a {\em packing LP}, consisting of  non-negative variables with each constraint of the form $\mathbf{a}^T x \le b$ where all entries in $\mathbf{a}$ and $b$ are non-negative. So we can solve \lp using faster approximation algorithms~\cite{PST91,GK07} for packing LPs, that also require only an approximate dual separation oracle (which is \kdo in our setting).

\paragraph{Rounding the LP solution.} Let $(x,y)$ denote the $O(\alpha)$-approximate solution to \lp. By Claim~\ref{cl:LP-obj} the objective value is $\Omega(\Opt/\alpha)$. The rounding algorithm below describes a (randomized) non-adaptive policy.
\begin{enumerate}
\item \label{step:corr1} For each $j\in[L]$, independently select solution $\tau_j\in \I_j$ as:
$$\tau_j \gets \left\{
\begin{array}{ll}
T&\mbox{ with probability } \frac{x^j_T}{2}, \mbox{ for each } T\in \I_j \\
\langle v_{j-1},v_j\rangle & \mbox{ with the remaining probability } 1-\frac12 \sum_{T\in \I_j} x^j_T
\end{array}\right.
$$

\item \label{step:corr2} If any vertex $v\in V$ appears in more than one solution $\{\tau_j\}_{j=0}^L$, then drop $v$ from all of them.

\item \label{step:corr3} If solution $\{\tau_j\}_{j=0}^L$ exceeds any constraint~\eqref{LP:4} by a factor more than $\frac{\log L}{\log\log L}$ then return an empty solution.

\item \label{step:corr4} For each $j\in[L]$, if $\tau_j$ contains multiple copies of any vertex $v\in V$ then retain only the copy with earliest deadline (i.e. highest reward).

\item \label{step:corr5} Return the non-adaptive policy that traverses the path $\tau_0\cdot \tau_1\cdots\tau_L$  and attempts each vertex independently with probability $\frac{\log\log L}{4K\cdot \log L}$.
\end{enumerate}

\paragraph{Analysis.} We now show that the expected reward of this non-adaptive policy is $\Omega(\frac{\log\log L}{\alpha\, K\cdot \log L})\cdot \Opt$, which would prove Theorem~\ref{thm:corr-NA}; recall that $K=\Theta(\log\log B)$ and $L=\Theta(\log B)$.

\begin{lemma} \label{lem:RRalt}
For any $j\in[L]$, $v\in V$ and $t\in I^j_v$,
$\Pr[\langle v,j,t\rangle \in \tau_j \mbox{ after Step~\ref{step:corr3}}]\,\ge\, y^j_{v,t}/8$.
\end{lemma}
\begin{proof}
Let $\tau^1_j$, $\tau^2_j$ and $\tau^3_j$ denote the solution $\tau_j$ after Step 1, 2 and 3 respectively. Clearly,
\vspace{-2mm}\begin{equation}\label{eq:prob-tau1}
\Pr\left[\langle v,j,t\rangle \in \tau_j^1\right] \quad = \quad\sum_{T\ni \langle v,j,t\rangle} \frac{x^j_T}{2} \quad = \quad \frac{y^j_{v,t}}{2}.
\end{equation}
\vspace{-3mm}

The last equality uses Claim~\ref{cl:xy-equality}. Note that $\langle v,j,t\rangle$ gets dropped in Step~\ref{step:corr2} exactly when there is some other solution $\{\tau^1_{\ell} : \ell\in[L]\setminus j\}$ that contains a copy of $v$. By union bound,~\eqref{eq:prob-tau1} and~\eqref{LP:3},
\vspace{-2mm}\begin{equation}\label{eq:prob-tau2}\Pr\left[\langle v,j,t\rangle \not\in \tau_j^2 \, |\, \langle v,j,t\rangle \in \tau_j^1\right] \quad \le \quad \sum_{\ell\in [L]\setminus j}\,\,  \sum_{t\in I^\ell_v} \frac{y^j_{v,t}}{2}\quad \le \quad \frac{1}{2}.
\end{equation}\vspace{-5mm}

In Step~\ref{step:corr3}, the entire solution is declared empty if any constraint~\eqref{LP:4} is violated by more than a factor of $\frac{\log L}{\log\log L}$; otherwise, $\tau^3_j=\tau^2_j$. Consider the constraint~\eqref{LP:4} with index $h\in[L]$. This reads $\sum_{\ell=0}^h \mu^h(\tau_\ell) \,\le\, (K+1)\cdot 2^h$, where $\mu^h(\tau_\ell):= \sum_{\langle v,\ell,t\rangle \in \tau_\ell} \mu^h_v$. The key observation is the following:
\begin{equation}\label{eq:main-RR}
\mbox{For any $\ell\le h$ and $\tau_\ell\in \I_\ell$, we have }\mu^h(\tau_\ell)\le (K+1)\cdot 2^h.
\end{equation}
This uses the fact that $\tau_\ell$ satisfies the knapsack constraint $\mu^\ell(\tau_\ell)\le (K+1)\cdot 2^\ell$ in \kdo instance $\I_\ell$; and by the observation~\eqref{eq:capped-mu} on capped sizes, $\frac{\mu^h(\tau_\ell)}{2^h}\le \frac{\mu^\ell(\tau_\ell)}{2^\ell}$ since $h\ge \ell$.

Using~\eqref{eq:main-RR} it follows that $Z_h:=\sum_{\ell=0}^h \frac{\mu^h(\tau_\ell)}{(K+1)\cdot 2^h}$ is the sum of independent $[0,1]$ bounded random variables.  Using~\eqref{eq:prob-tau1}, and \lp constraint~\eqref{LP:4} we have:
$$\E[Z_h] \quad \le \quad \frac1{(K+1)\cdot 2^h} \,\,\sum_{\ell=0}^h \sum_{v\in V} \sum_{t\in I^\ell_v} \,\mu^h_v\cdot \frac{y^j_{v,t}}{2} \quad \le \quad \frac12.$$
Hence, by a Chernoff bound, {\small $\Pr\left[Z_h > \frac{\log L}{\log\log L}\right]  \le \frac{1}{2L}$}.

Taking a union bound over all $L$ constraints~\eqref{LP:4}, it follows that with probability at least half, none of them is violated by a factor more than $\frac{\log L}{\log\log L}$. That is, $\Pr\left[\langle v,j,t\rangle \in \tau_j^3 \, |\, \langle v,j,t\rangle \in \tau_j^2\right] \ge \frac12$.

Combined with~\eqref{eq:prob-tau1} and \eqref{eq:prob-tau2}, we obtain the lemma.
\end{proof}

\begin{claim}\label{cl:corr-exp}
The expected \lp objective value of solution $\{\tau_j\}_{j=0}^L$ after Step~\ref{step:corr4} is $\Omega(\Opt/\alpha)$.
\end{claim}
\begin{proof}
By Lemma~\ref{lem:RRalt} it follows that expected \lp objective value of solution $\{\tau_j\}_{j=0}^L$ after Step~\ref{step:corr3} is at least:
{\small $$\sum_{j=0}^L \, \sum_{v \in V}  \, \sum_{t \in I^j_v} \, f^j_v(t)\cdot \frac{y^j_{v,t}}8 \quad \ge\quad \Omega\left(\frac{\Opt}{\alpha}\right).$$}
The last inequality uses Claim~\ref{cl:LP-obj} and the fact that we have an $O(\alpha)$-approximately optimal \lp solution.

In Step~\ref{step:corr4}, we retain only one copy of each vertex in each $\tau_j$. Using Claim~\ref{cl:gknr-count}, since we retain the most profitable copy of each vertex, this decreases the total reward of each $\tau_j$ by at most a factor of $3$. %This completes the proof.
\end{proof}

%We are now ready to analyze the non-adaptive
Consider now the non-adaptive policy in Step~\ref{step:corr5} and {\em condition} on any solution $\{\tau_j\}_{j=0}^L$. Fix any $j\in [L]$ and  $\langle v, j, t\rangle\in \tau_j$. The distance traveled until $v$ is at most $\sum_{i=0}^{j-1}D_i + t$. By Step~\ref{step:corr4} and \lp constraint~\eqref{LP:4}, the total $\mu^j$-size of vertices in $\{\tau_i\}_{i=0}^j$ is at most $\frac{(K+1)\log L}{\log\log L}\cdot 2^j$. Since each vertex is attempted only with probability $\frac{\log\log L}{4K\cdot \log L}$, we have
$$\Pr\left[ \sum_{u\prec v} \min\{S_u,2^j\} \ge 2^j \right]  \,\, \le \,\, \frac12,$$ where the summation ranges over all vertices visited before $v$. Thus with probability at least half, vertex $v$ is visited  by time $\sum_{i=0}^{j-1}D_i + t+2^j-1$.

So the expected reward from vertex $v$ is $\Omega\left(\frac{\log\log L}{K\cdot \log L}\right)\cdot \eta_v(B-\sum_{i=0}^{j-1}D_i - t- 2^j+1) = \Omega\left(\frac{\log\log L}{K\cdot \log L}\right)\cdot f^j_v(t)$.
Adding this contribution over all vertices, the total reward is $\Omega\left(\frac{\log\log L}{K\cdot \log L}\right)$ times the \lp objective of solution $\{\tau_j\}_{j=0}^L$ after Step~\ref{step:corr4}. Taking expectations over Steps~\ref{step:corr1}-\ref{step:corr4}, and using Claim~\ref{cl:corr-exp}, it follows that our non-adaptive policy has expected reward $\Omega\left(\frac{\log\log L}{\alpha\, K\cdot \log L}\right)\cdot \Opt$. This completes the proof of Theorem~\ref{thm:corr-NA}.

\section{Conclusion}
In this paper, we proved an $\Omega(\sqrt{\log\log B})$ lower bound on the adaptivity gap of stochastic orienteering. The best known upper bound  is $O(\log\log B)$~\cite{GKNR12}. Closing this gap is an interesting open question. For the {\em correlated} stochastic orienteering problem, we gave a quasi-polynomial time $O(\alpha\cdot \log^2\log B)$-approximation algorithm, where $\alpha$ denotes the best approximation ratio for the deadline-orienteering problem. It is known that correlated stochastic orienteering can not be approximated to a factor better than $\Omega(\alpha)$~\cite{GKNR12}. Finding an $O(\alpha)$ approximation algorithm for correlated stochastic orienteering is another interesting direction.

\bibliographystyle{plain}
\bibliography{stoc-ks}

\end{document}